\documentclass[12pt,twoside]{article}
\usepackage{amsfonts}
\usepackage{cite}
 \usepackage{color}
\usepackage{amsmath} \numberwithin{equation}{section}
\usepackage{amsthm}
 \usepackage{graphicx}
\usepackage[titletoc,title]{appendix}
\usepackage{subfigure}

\topmargin by-\headsep  \textheight 24cm
\newtheorem{lemma}{Lemma}[section]

\newtheorem{definition}{Definition}[section]
\newtheorem{theorem}[definition]{Theorem}

\newtheorem{remark}{Remark}[section]

\newcommand{\gtheta}{\frac{\Theta(t)}{1+\alpha \Theta^2(t)}}
\newcommand{\gthetaa}{\frac{\Theta}{1+\alpha \Theta^2}}
\newcommand{\gthetaaa}{\frac{\Theta^*}{1+\alpha (\Theta^*)^2}}

\newcommand{\mk}[1][k]{\ensuremath{\langle #1 \rangle}}
\begin{document}
\title{ Global stability of a network-based SIRS epidemic model with nonmonotone incidence rate}%
\author{Lijun Liu$^{1}$ \thanks{Corresponding author. \newline E-mail address:   manopt@163.com (L. Liu)}, \quad Xiaodan Wei$^{2}$, \quad Naimin Zhang$^3$ \\
  \small  1. Department of Mathematics, Dalian Nationalities University, Dalian
   116600,  China\\
\small 2. College of Computer  Science, Dalian
  Nationalities University, Dalian  116600,  China\\
\small 3.
School of Mathematics and Information Science, Wenzhou University,
Wenzhou 325035, China
   }

\pagestyle{myheadings}
\markboth{ }%
{}

\date{}
 \maketitle

\begin{abstract}
This paper studies the dynamics of a network-based SIRS epidemic model with vaccination and a nonmonotone incidence rate. This type of nonlinear incidence can be used to describe
the psychological or inhibitory effect from the behavioral change of the susceptible individuals when the number of infective individuals on heterogeneous networks is getting larger. Using the analytical method, epidemic threshold  $R_0$ is obtained. When $R_0$ is less than one, we prove the disease-free equilibrium  is globally asymptotically stable and the disease dies out, while  $R_0$ is greater than one, there exists a unique endemic equilibrium. By constructing a suitable Lyapunov function, we also prove the endemic equilibrium is globally asymptotically stable if the inhibitory factor $\alpha$ is sufficiently large.  Numerical experiments are also given to support the theoretical results. It is shown both theoretically and numerically a larger $\alpha$ can accelerate the extinction of the disease and reduce the level of disease.

\medskip


 \end{abstract}


%

\section{Introduction}
Mathematical models which describe the dynamics of infectious diseases have played a crucial role in the disease control in epidemiological aspect. In the traditional epidemiology, it is commonly assumed
 that individuals mix uniformly and all hosts have identical rates of disease-causing contacts. This over-simplified assumption makes the analysis tractable but not realistic\cite{bansal2007when}. Virtually, the interpersonal contact underlying disease transmission can be thought of expanding a complex network, where relations (edges) join individuals (nodes) who interact with each other \cite{Pastor-Satorras20013200, Mao-Xing20092115}. Therefore, the disease transmission should be modeled over complex networks. Related studies indicated that networks structures have profound impacts on the spreading dynamics
\cite{peng2013vaccination,Newman2002,Wang2012543,Zhu20122588}.

 One of the pioneer works in
this area  was done by Pastor-Satorras and Vespignani
\cite{Pastor-Satorras2002,Pastor-Satorras20013200},
where they first succeeded in studying susceptible-infectious-susceptible (SIS) epidemic model on scale-free networks by large-scale simulations. The most striking result is that they found the absence of the epidemic threshold in these networks. That is, the threshold approaches zero in the limit of a large number of edges and nodes, and even a quite small infectious rate can produce a major epidemic outbreak,  for
which the rigorous proof was given by Wang and Dai \cite{Wang20081495} by a monotone iterative technique.  Moreno et al. \cite{Moreno2002521} studied susceptible-infectious-recovered (SIR) epidemic model on scale-free complex population networks. They show that the large connectivity fluctuations usually
found in these networks strengthen considerably the incidence of epidemic outbreaks. These outstanding results have inspired a great number of related works (see Refs. \cite{Li20141042,Liu20113375, Li2014686, ZHU2017614, Huang2017296, Wei2017a} and the references therein).

It is well-known that the spread of many human diseases can be prevented or reduced by vaccination of the susceptible individuals. Chen and Sun \cite{Chen2014196} firstly succeeded in studying optimal control of an susceptible-infectious-recovered-susceptible (SIRS) epidemic model with vaccination on heterogeneous networks. They showed that if the percentage of vaccination of the susceptible is smaller than the recovered rate, the diseases may persist on heterogeneous networks.

We shall emphasize that the incidence rate plays an important role in guaranteeing that the model can give a reasonable approximative description for the disease dynamics. However, in most existing epidemic models on complex networks, the incidence rate is usually assumed to be bilinear function based on the mass action law for infection. In fact, there are several reasons for using nonlinear incidence rates, even nonmonotone incidence function\cite{Xiao2007419, Zhang201424, Yuan2012, Yang2010, Muroya2011}. In practical situations, the number of effective contacts between infective individuals and susceptible individuals usually  decreases at high infective levels due to the quarantine of infective individuals or due to the protection measures by the susceptible individuals. Very recently, for modelling such a psychological behavior on complex networks, Li \cite{Li2015234} studied  the dynamics of a network-based SIS epidemic model with nonmonotone incidence rate,
\begin{equation}\label{eq:li-sis}
\left
\{\begin{split}
&\frac{d S_k(t)}{dt}=-\lambda kS_k(t)g(\Theta)+I_k(t), \\[2mm]
&\frac{d I_k(t)}{dt}= \lambda kS_k(t)g(\Theta)-I_k(t),\quad k=1,2,\cdots,n,
\end{split}
\right.
\end{equation}
where $\lambda>0$ is the transmission rate when susceptible individuals contact with infectious. $S_k(t)$ and $I_k(t)$ denote the relative densities of susceptible and infectious with degree $k$ at time $t$ on the complex networks with maximum degree $n$. The connectivity of nodes on the network is assumed to be uncorrelated, thus, we have $\Theta(t)=\frac{1}{\langle k
\rangle}\sum_{k=1}^nkP(k)I_k(t)$ with $\langle
k\rangle=\sum_{k=1}^n kP(k)$ is the average degree of the network and $P(i)$ is the connectivity distribution.

The function $g(\Theta)$ is introduced in SIS model \eqref{eq:li-sis} to reflect the psychological behavior mentioned above, which leads to nonlinear incidence rate defined by
\begin{equation}\label{nonlinear incidence}
  \lambda kS_k(t)g(\Theta):=\lambda kS_k(t)\gthetaa,
\end{equation}
where $\alpha\geq 0$ is a parameter.

When $\alpha=0$, then the nonlinear incidence rate \eqref{nonlinear incidence} becomes the bilinear one. Hence, the SIS model with \eqref{nonlinear incidence} can be seen a generalization of the existing SIS model. Meanwhile, if $\alpha$ is large enough (e.g., $\alpha>2$), the function $g$ becomes a nonmonotone function. The biological meaning is that at high infective risk (i.e., when $\Theta$ is sufficiently large),
the incidence rate may decrease as $\Theta$ increases because individuals become more careful and tend to reduce their contacts with other ones. In this sense, we call the parameter $\alpha$ inhibitory factor from the behavioral change of the susceptible individuals.

Inspired by the works of \cite{Chen2014196} and \cite{Li2015234}, in this paper we propose the following SIRS epidemic model with vaccination and the nonmonotone incidence rate as follows,
\begin{equation}\label{eq:nm-sirs-simplify}
\left\{\begin{split}
&\frac{d S_k(t)}{dt}=-\lambda kS_k(t)\gtheta+\delta R_k(t)-\mu S_k(t), \\[2mm]
&\frac{d I_k(t)}{dt}= \lambda kS_k(t)\gtheta-\gamma I_k(t),\\
&\frac{d R_k(t)}{dt}=\gamma I_k(t)-\delta R_k(t)+\mu S_k(t),\quad k=1,2,\cdots,n.
\end{split}\right.
\end{equation}
where $S_k(t), I_k(t), R_k(t)$ denote the relative densities of susceptible nodes, infectious nodes and recovered nodes with degree $k$ respectively. $\delta>0$ represents the rate of immunization-lost for recovered nodes; $\gamma>0$ represents the recovery rate of infected nodes; $\mu>0$ represents the vaccination percentage for the susceptible nodes.

Among the existing epidemic models on complex networks, there has far been relatively little research into network epidemic models with nonlinear incidence rate.
In \cite{Li2015234}, the author successfully proposed the SIS model with nonmonotone incidence rate and obtained the epidemic threshold.  It was proved that if
the transmission rate is below the threshold, the disease-free equilibrium is globally
asymptotically stable, otherwise the endemic equilibrium is permanent. In our recent paper \cite{Wei2017}, the global stability and attractivity of the endemic equilibrium  of system \eqref{eq:li-sis} is rigourously proved. When $\alpha=0$, the proposed model \eqref{eq:nm-sirs-simplify} can be simplified to the model proposed by Chen and Sun \cite{Chen2014196}. In this paper, the global stability of disease-free equilibrium as well as the endemic equilibrium is rigourously proved without additional assumptions on the constants, which improves the existing results.


The remainder of this paper is organized as follows. In Section 2, we show that the solutions of system \eqref{eq:nm-sirs-simplify} are positive
and the epidemic threshold is obtained. In Section 3, we study the global stability of the disease-free equilibrium as well as the endemic equilibrium. In Section 4, numerical experiments are given to illustrate the theoretical results. Finally, conclusions  are drawn in Section 5.

\section{Positivity of solutions and the epidemic threshold}
From a practical perspective, the initial conditions for system \eqref{eq:nm-sirs-simplify} satisfy
\begin{equation}\label{eq:intial conditions}
\begin{split}
    I_k(0)\geq 0,\; R_k(0)\geq 0, \;
    S_k(0)=1-I_k(0)-R_k(0)>0,\; k=1,2,\ldots,n,\; \Theta(0)>0.
\end{split}
\end{equation}
Note that  $\frac{d}{dt}(S_k(t)+I_k(t)+R_k(t))=0$ and $S_k(t)+I_k(t)+R_k(t)=1$ for all $t\geq 0$ and for all
$k=1,2,\ldots,n.$ So the system \eqref{eq:nm-sirs-simplify} becomes the following form:
\begin{equation}\label{eq:nm-si-system}
\left\{\begin{split}
&\frac{d S_k(t)}{dt}=-\lambda kS_k(t)\gtheta+\delta (1-S_k(t)-I_k(t))-\mu S_k(t), \\[2mm]
&\frac{d I_k(t)}{dt}= \lambda kS_k(t)\gtheta-\gamma I_k(t),\quad k=1,2,\cdots,n.
\end{split}\right.
\end{equation}%

As system \eqref{eq:nm-si-system} is equivalent to the system \eqref{eq:nm-sirs-simplify}, from now on we only consider the dynamics of \eqref{eq:nm-si-system}.

First we establish the positivity of solutions of system \eqref{eq:nm-si-system} in the following lemma.
\begin{lemma}\label{lm:1}
  Let $(S_1,I_1,\ldots,S_n,I_n)$ be the solution of SIRS system \eqref{eq:nm-si-system} with initial conditions \eqref{eq:intial conditions}. Then for $k = 1, 2, \ldots, n$, we have $0 < S_k (t) < 1, 0 < I_k (t) < 1$, and $0 < \Theta(t) < 1$ for all $t > 0$.
\end{lemma}
The proof is shown in Appendix A.

Now we are going to compute all biologically feasible equilibria $(S_k, I_k)$ admitted by system \eqref{eq:nm-si-system}.

First, it is easy to see the disease-free equilibrium of \eqref{eq:nm-si-system} is $E^0=
(\frac{\delta}{\delta+\mu},0,\ldots,\frac{\delta}{\delta+\mu},0)$.

 In order to calculate
the epidemic equilibrium $E^*(S_1^*,I_1^*,\ldots,S_n^*,I_n^*)$. Let $S_k^\prime(t)=0$ and $I_k^\prime(t)=0$. It  follows from \eqref{eq:nm-si-system} that

\newcommand{\DF}{\gamma\delta+\alpha\delta\gamma\Theta^2}
\newcommand{\DS}{\lambda\delta k\Theta}
\newcommand{\DD}{\gamma(\delta+\mu)+\lambda(\delta+\gamma)k\Theta+\alpha\gamma(\delta+\mu)\Theta^2}

\newcommand{\DFS}{\gamma\delta+\alpha\delta\gamma(\Theta^*)^2}
\newcommand{\DSS}{\lambda\delta k\Theta^*}
\newcommand{\DDS}{\gamma(\delta+\mu)+\lambda(\delta+\gamma)k\Theta^*+\alpha\gamma(\delta+\mu)(\Theta^*)^2}

\begin{equation}\label{eq:epi-equilirium}
  \left\{\begin{split}
      &S_k =\frac{\DF}{\DD}, \\
      &I_k=\frac{\DS}{\DD},\quad k=1,2,\ldots,n.
  \end{split}\right.
\end{equation}
where $\Theta=\frac{1}{\langle k
\rangle}\sum\limits_{i=1}^niP(i)I_i$. From \eqref{eq:epi-equilirium}, we obtain the self-consistency equation of the form
$\Theta f(\Theta)=\Theta$, where
\begin{equation}\label{eq:self-cons-theta}
  f(\Theta) \equiv \frac{1}{\langle k \rangle}\sum_{k=1}^{n}\frac{\lambda\delta k^2P(k)}{\DD}.
\end{equation}

\noindent Since $f^\prime(\Theta)<0$ and $f(1)<1$, the equation $\Theta f(\Theta)=\Theta$ has a unique non-trivial solution
 if and only if $f(0)>1$, that is
 \begin{equation}\label{eq1}
   \frac{\lambda\delta \langle k^2 \rangle}{\gamma(\delta+\mu)\langle k \rangle}>1,
 \end{equation}
 where $\langle k^2 \rangle=\sum_{k=1}^{n}k^2P(k)$. These analyses lead to the following result.

 \begin{lemma}\label{lm:2}
   Define the epidemic threshold
   \begin{equation*}
     R_0:=\frac{\lambda\delta \langle k^2 \rangle}{\gamma(\delta+\mu)\langle k \rangle}.
   \end{equation*}
   If $R_0>1$, then system \eqref{eq:nm-si-system} admits a unique positive equilibrium $E^*(S_1^*,I_1^*,\ldots,S_n^*,I_n^*)$, which corresponds to the endemic equilibrium of system \eqref{eq:nm-sirs-simplify} and satisfies
\begin{equation}\label{eq:epi-equilirium2}
  \left\{\begin{split}
      &S_k^* =\frac{\DFS}{\DDS}, \\
      &I_k^*=\frac{\DSS}{\DDS},\quad k=1,2,\ldots,n.
  \end{split}\right.
\end{equation}
where $\Theta^*=\langle k\rangle^{-1}\sum_{k=1}^n kP(k)I_k^* \in (0, 1)$ is the unique positive root of the equation
$$\Theta=\frac{1}{\langle k \rangle}\sum_{k=1}^{n}\frac{\lambda\delta k^2P(k)\Theta}{\DD}.$$

 \end{lemma}
\begin{remark}
  From Lemma \ref{lm:2}, we can see that the epidemic threshold is determined in terms of the network structure and this
threshold is same as that derived in Ref. \cite{Chen20114391}. In other words, the nonlinear incidence rate does not affect the
threshold $R_0$. Besides, as the result obtained in Ref. \cite{Pastor-Satorras20013200}, the spreading processes of our model do not possess an epidemic
threshold in an infinite scale-free network since $\langle k^2 \rangle \to \infty$ in this situation.
\end{remark}
\section{Global stability analysis}
In this section, we investigate the globally asymptotical stability of the disease-free equilibrium $E^0$ as well as the endemic equilibrium $E^*$.
\begin{theorem}
  Let $R_0:=\frac{\lambda\delta \langle k^2 \rangle}{\gamma(\delta+\mu)\langle k \rangle} < 1$. Then the disease-free equilibrium $E^0=
(\frac{\delta}{\delta+\mu},0,\ldots,\frac{\delta}{\delta+\mu},0)$ of system \eqref{eq:nm-si-system} is globally asymptotically
stable, i.e., the disease fades out.
\end{theorem}
\begin{proof}
We first claim that
\begin{equation}\label{eq3}
  \limsup_{t\to +\infty} S_k(t)=\frac{\delta}{\delta+\mu}:=S_k^0.
\end{equation}
Since by Lemma \ref{lm:1}, we have $0<S_k(t), I_k(t) <1$,  and $0<\Theta(t)<1$ for all
$t>0$.
From the first equation of system \eqref{eq:nm-si-system}, it follows that
  \begin{equation}\label{eq4}
    \frac{dS_k(t)}{dt}\leq \delta-(\delta+\mu)S_k(t).
  \end{equation}
which implies \eqref{eq3}. So for any $\epsilon>0,$ there holds $S_k(t)\leq \frac{\delta}{\delta+\mu}+\epsilon=S_k^0+\epsilon$ when $t$ is sufficiently large. Thus from the second equation of system \eqref{eq:nm-si-system}, we have
\begin{equation}\label{eq5}
  \frac{dI_k(t)}{dt} < \lambda k (S_k^0+\epsilon)\Theta(t)-\gamma I_k(t).
\end{equation}
It then suffices to show the positive solutions of the system
\begin{equation}\label{eq6}
  \frac{dI_k(t)}{dt} = \lambda k (S_k^0+\epsilon)\Theta(t)-\gamma I_k(t)
\end{equation}
tend to zero at $t$ goes to infinity. Consider the Lyapunov function
\begin{equation}\label{eq7}
  V(t)=\sum_{k=1}^{n}w_k I_k(t),
\end{equation}
where
$$w_k=\frac{kP(k)}{\gamma\langle k\rangle},\quad k=1,2,\ldots,n.$$
Calclulating the derivative of $V(t)$ along the positive solutions of \eqref{eq6}, we have
\begin{equation}\label{eq8}
  \begin{split}
     \frac{dV(t)}{dt} &=\sum_{k=1}^{n} w_k [\lambda k(S_k^0+\epsilon)\Theta(t)-\gamma I_k(t)], \\
       & =\sum_{k=1}^{n} \frac{kP(k)}{\gamma \langle k \rangle}[\lambda k(S_k^0+\epsilon)\Theta(t)-\gamma I_k(t)],\\
       & =\Theta(t)\left(R_0+\frac{\lambda \langle k^2 \rangle \epsilon}{\gamma \langle k \rangle}-1\right).
  \end{split}
\end{equation}
Since $R_0 <1$, we can fix an $\epsilon>0$ small enough such that $R_0+\frac{\lambda \langle k^2 \rangle \epsilon}{\gamma \langle k \rangle} < 1$. That ensures $\frac{dV}{dt} \leq 0$ and $\frac{dV}{dt}=0$ holds only if $I_k=0$ for $k=1,\cdots,n$. Hence the trivial solution of \eqref{eq6} is globally asymptotically stable. Further from the nonnegativeness of the positive solution of \eqref{eq:nm-si-system} and the the comparison theorem, we complete the proof.
\end{proof}

Now our task is to claim the globally asymptotical stability of the endemic equilibrium.

\begin{theorem}\label{th2}Let $R_0:=\frac{\lambda\delta \langle k^2 \rangle}{\gamma(\delta+\mu)\langle k \rangle} > 1$.
If $\alpha\geq \alpha_c:=\frac{\lambda^2 n^2(R^0)^2}{16\omega^2}$ or $\alpha\leq \alpha_l:=\frac{2\omega}{\lambda n}$ with constant $\omega>0$ defined in equation \eqref{eq:B5}, then the
endemic equilibrium $E^*(S_1^*,I_1^*,\ldots,S_n^*,I_n^*)$ of
system \eqref{eq:nm-si-system} is globally asymptotically stable.
\end{theorem}
\begin{proof}
By Lemma \ref{lm:1}, we have $0<S_k(t)<1, 0<\Theta(t)<1$ for all
$t>0$. The first equation of \eqref{eq:nm-si-system} can be reformulated as
\begin{equation}\label{e2}
  \begin{split}
 &S_k'(t)=\delta-(\delta+\mu)S_k(t)-\delta I_k(t)-\lambda k S_k(t)\frac{\Theta(t)}{1+\alpha\Theta^2(t)}.
 \end{split}
 \end{equation}
Moreover, one can derive from the second equation of  \eqref{eq:nm-si-system} that
\begin{equation}\label{eu4}
\begin{split}
 \Theta'(t)&=\frac{1}{\mk[k]}\sum_{k=1}^{n}kP(k)I_k^\prime(t)\\
           &=\frac{1}{\mk[k]}\sum_{k=1}^{n}kP(k)\left[\lambda k S_k(t)\gtheta-\gamma I_k(t)\right]\\
          &=-\gamma \Theta(t)+\frac{\lambda}{\mk}\sum_{k=1}^{n}k^2P(k)S_k(t)\gtheta.
\end{split}\end{equation}
Consider the following  Lyapunov function
\begin{equation*}\label{eu5}
\begin{split}
V(t)=\frac12\sum\limits_{k=1}^{n}\left\{a_k\big(S_k(t)-S_k^*\big)^2+m a_k(R_k(t)-R_k^*)^2\right\}+\Theta(t)-\Theta^*-\Theta^*\ln\frac{\Theta(t)}{\Theta^*},
\end{split}\end{equation*}
 where
\begin{equation*}\label{eu5}
\begin{split}
 R_k^*=1-S_k^*-I_k^*,\quad a_k=\frac{kP(k)}{\langle k \rangle S_k^*}, \quad k=1,2,\cdots,n,
\end{split}\end{equation*}
and the positive constant $m$ will be determined later.

Calculating the derivative of $V(t)$ along the positive solution of \eqref{eq:nm-si-system}, we
have
\begin{equation}\label{eu66}\begin{split}
V'=\sum\limits_{k=1}^{n}a_k(S_k-S_k^*) S_k'+\sum\limits_{k=1}^{n}m a_k(R_k-R_k^*) R_k'+\frac{\Theta-\Theta^*}{\Theta}\Theta'=:V_1+V_2+V_3.\\[2mm]
\end{split}\end{equation}

Using \eqref{eu4} and the identity  $1=\frac{\lambda}{\gamma \langle k \rangle}\sum\limits_{k=1}^{n}
\frac{k^2P(k)S_k^*}{1+\alpha(\Theta^*)^2}$, we obtain
\begin{equation}\label{eu660}\begin{split}
 V_3
=&\gamma(\Theta-\Theta^*)\left(- 1+\frac{\lambda}{\gamma\langle k \rangle}\sum\limits_{k=1}^{n}\frac{k^2P(k)S_k}{1+\alpha\Theta^2}\right)\\[2mm]
=&(\Theta-\Theta^*)\frac{\lambda}{\langle k \rangle}\sum\limits_{k=1}^{n}k^2P(k)\left(\frac{S_k}{1+\alpha\Theta^2}-\frac{S_k^*}{1+\alpha(\Theta^*)^2}\right)\\[2mm]
=&\frac{\lambda}{\langle k \rangle}\sum\limits_{k=1}^{n}k^2P(k)\left[\frac{(S_k-S_k^*)(\Theta-\Theta^*)}{1+\alpha(\Theta^*)^2} - \frac{\alpha S_k(\Theta-\Theta^*)(\Theta^2-(\Theta^*)^2)}{(1+\alpha(\Theta^*)^2)(1+\alpha\Theta^2)}\right].\\[2mm]
\end{split}\end{equation}

Using the last equation of \eqref{eq:nm-sirs-simplify} and the identity $\gamma I_k^*-\delta R_k^*+\mu S_k^*=0$, we have
\begin{equation}\label{eqv2}
  \begin{split}
     V_2 &=\sum_{k=1}^{n}  ma_k(R_k-R_k^*)\left(\gamma I_k-\delta R_k+\mu S_k\right) \\
       & =\sum_{k=1}^{n} m a_k(R_k-R_k^*)\big[\gamma (I_k-I_k^*)-\delta (R_k-R_k^*)+\mu (S_k-S_k^*)\big] \\
       & =\sum_{k=1}^{n} m a_k(R_k-R_k^*)\big[-\gamma (R_k-R_k^*)-\gamma(S_k-S_k^*)-\delta (R_k-R_k^*)+\mu (S_k-S_k^*)\big] \\
       & =-\sum_{k=1}^{n} (\gamma+\delta) m a_k(R_k-R_k^*)^2-\sum_{k=1}^{n} (\gamma-\mu)m a_k(R_k-R_k^*)(S_k-S_k^*).
  \end{split}
\end{equation}

Using \eqref{e2} and the identity $\delta=(\delta+\mu)S_k^*+\delta I_k^*+\lambda k S_k^*\gthetaaa$, we have
\begin{equation}\label{eu555}
\begin{split}
 V_1=&\sum\limits_{k=1}^{n}a_k(S_k-S_k^*)\left(\delta-(\delta+\mu)S_k-\delta I_k-\lambda k S_k\gthetaa\right)\\[2mm]
=&\sum\limits_{k=1}^{n}a_k(S_k-S_k^*)\left\{-(\delta+\mu)(S_k-S_k^*)-\delta(I_k-I_k^*)+\lambda k \left[\frac{S_k^*\Theta^*}{1+\alpha(\Theta^*)^2}-\frac{S_k\Theta}{1+\alpha\Theta^2}\right]\right\}\\[2mm]
=&\sum\limits_{k=1}^{n}a_k(S_k-S_k^*)\left\{-\mu(S_k-S_k^*)+\delta(R_k-R_k^*)+\lambda k \left[\frac{S_k^*\Theta^*}{1+\alpha(\Theta^*)^2}-\frac{S_k\Theta}{1+\alpha\Theta^2}\right]\right\}\\[2mm]
=&-\sum\limits_{k=1}^{n}\mu a_k(S_k-S_k^*)^2+\sum_{k=1}^{n}\delta a_k (S_k-S_k^*)(R_k-R_k^*) \\[2mm]
&-\frac{\Theta}{1+\alpha\Theta^2}\sum\limits_{k=1}^{n}\lambda k a_k(S_k-S_k^*)^2-\sum\limits_{k=1}^{n}\frac{\lambda k a_k S_k^*}{1+\alpha(\Theta^*)^2}(S_k-S_k^*)(\Theta-\Theta^*)\\[2mm]
&+\sum\limits_{k=1}^{n}\alpha\lambda k a_k S_k^*(S_k-S_k^*)\frac{\Theta(\Theta^2-(\Theta^*)^2)}{(1+\alpha(\Theta^*)^2)(1+\alpha\Theta^2)}.\\[2mm]
\end{split}\end{equation}
Note that the final term of right hand side of the above equality can be written as
\begin{equation*}
\begin{split}
 &\sum\limits_{k=1}^{n}\alpha\lambda k a_k S_k^*(S_k-S_k^*)\frac{\Theta(\Theta^2-(\Theta^*)^2)}{(1+\alpha(\Theta^*)^2)(1+\alpha\Theta^2)}\\[2mm]
 &=\sum\limits_{k=1}^{n}\alpha\lambda k a_k S_k^*(S_k-S_k^*)\left[\frac{(\Theta-\Theta^*)(\Theta^2-(\Theta^*)^2)}{(1+\alpha(\Theta^*)^2)(1+\alpha\Theta^2)}+\frac{\Theta^*(\Theta^2-(\Theta^*)^2)}{(1+\alpha(\Theta^*)^2)(1+\alpha\Theta^2)}\right]\\[2mm]
&=\sum\limits_{k=1}^{n}\alpha  \lambda k a_k S_k^*S_k\frac{ (\Theta-\Theta^*)(\Theta^2-(\Theta^*)^2)}{(1+\alpha(\Theta^*)^2)(1+\alpha\Theta^2)}-\sum\limits_{k=1}^{n}\frac{\alpha\lambda k a_k (S_k^*)^2  (\Theta+\Theta^*) }{(1+\alpha(\Theta^*)^2)(1+\alpha\Theta^2)}(\Theta-\Theta^*)^2\\[2mm]
&\quad+\sum\limits_{k=1}^{n}\frac{\alpha\lambda k a_k S_k^* \Theta^*(\Theta+\Theta^*) }{(1+\alpha(\Theta^*)^2)(1+\alpha\Theta^2)}(S_k-S_k^*)(\Theta-\Theta^*).\\[2mm]
\end{split}\end{equation*}

Substituting it into \eqref{eu555} yields
\begin{equation}\label{eu55}
\begin{split}
 V_1=&-\sum\limits_{k=1}^{n}\mu a_k(S_k-S_k^*)^2+\sum_{k=1}^{n}\delta a_k (S_k-S_k^*)(R_k-R_k^*)\\[2mm]
 &-\frac{\Theta}{1+\alpha\Theta^2}\sum\limits_{k=1}^{n}\lambda k a_k(S_k-S_k^*)^2\\[2mm]
&+\sum\limits_{k=1}^{n}\frac{\alpha\lambda k a_k S_k^* \Theta^*(\Theta+\Theta^*) }{(1+\alpha(\Theta^*)^2)(1+\alpha\Theta^2)}(S_k-S_k^*)(\Theta-\Theta^*)\\[2mm]
&-\sum\limits_{k=1}^{n}\frac{\alpha\lambda k a_k (S_k^*)^2  (\Theta+\Theta^*) }{(1+\alpha(\Theta^*)^2)(1+\alpha\Theta^2)}(\Theta-\Theta^*)^2\\[2mm]
&-\sum\limits_{k=1}^{n}\lambda k a_k S_k^*\left[\frac{(S_k-S_k^*)(\Theta-\Theta^*)}{1+\alpha(\Theta^*)^2} -\frac{\alpha S_k (\Theta-\Theta^*)(\Theta^2-(\Theta^*)^2)}{(1+\alpha(\Theta^*)^2)(1+\alpha\Theta^2)}\right].\\[2mm]
\end{split}\end{equation}
Since $a_k=\frac{kP(k)}{\langle k \rangle S_k^*}$, we find the final term of \eqref{eu55} is equal to $-V_3$. Then combining \eqref{eqv2}, it follows that
\begin{equation}\label{eu9}
\begin{split}
V'=&-\sum\limits_{k=1}^{n}a_kF_k(m)-\frac{\Theta}{1+\alpha\Theta^2}\sum\limits_{k=1}^{n}\lambda k a_k(S_k-S_k^*)^2\\[2mm]
&+\sum\limits_{k=1}^{n}\frac{\alpha\lambda k a_k S_k^* \Theta^*(\Theta+\Theta^*) }{(1+\alpha(\Theta^*)^2)(1+\alpha\Theta^2)}(S_k-S_k^*)(\Theta-\Theta^*)\\[2mm]
&-\sum\limits_{k=1}^{n}\frac{\alpha\lambda k a_k (S_k^*)^2  (\Theta+\Theta^*) }{(1+\alpha(\Theta^*)^2)(1+\alpha\Theta^2)}(\Theta-\Theta^*)^2,
\end{split}\end{equation}
where
\begin{equation}\label{eq:Fk(M)}
  F_k(m):=\mu(S_k-S_k^*)^2+\big[(\gamma-\mu)m-\delta\big](R_k-R_k^*)(S_k-S_k^*)+(\gamma+\delta)m(R_k-R_k^*)^2.
\end{equation}
In the Appendix B, we show there exists $m^*>0$ such that
\begin{equation}\label{eq:inequ-Fk(M)}
  F_k(m^*)\geq \omega(m^*)(S_k-S_k^*)^2+\nu(m^*)(R_k-R_k^*)^2,
\end{equation}
where $\omega(m^*), \nu(m^*)>0$ are two constants shown in \eqref{eq:B5} and \eqref{eq:B6}. For simplicity, in the following, we denote $\omega(m^*)$ by $\omega$, $\nu(m^*)$ by $\nu$ respectively.

Substituting \eqref{eq:inequ-Fk(M)} into \eqref{eu9}, we have
\begin{equation}\label{eu11}
\begin{split}
V'\leq &-\sum\limits_{k=1}^{n}\omega a_k(S_k-S_k^*)^2-\sum\limits_{k=1}^{n}\nu a_k(R_k-R_k^*)^2-\frac{\Theta}{1+\alpha\Theta^2}\sum\limits_{k=1}^{n}\lambda k a_k(S_k-S_k^*)^2\\[2mm]
&+\sum\limits_{k=1}^{n}\frac{\alpha\lambda k a_k S_k^* \Theta^*(\Theta+\Theta^*) }{(1+\alpha(\Theta^*)^2)(1+\alpha\Theta^2)}(S_k-S_k^*)(\Theta-\Theta^*)\\[2mm]
&-\sum\limits_{k=1}^{n}\frac{\alpha\lambda k a_k (S_k^*)^2  (\Theta+\Theta^*) }{(1+\alpha(\Theta^*)^2)(1+\alpha\Theta^2)}(\Theta-\Theta^*)^2
\\[2mm]
=&-\frac{\Theta}{1+\alpha\Theta^2}\sum\limits_{k=1}^{n}\lambda k a_k(S_k-S_k^*)^2-\sum_{k=1}^{n}
\nu a_k (R_k-R_k^*)^2\\[2mm]
&-\sum\limits_{k=1}^{n}\omega a_k\left\{X_k^2-\frac{\alpha\lambda k S_k^* \Theta^*(\Theta+\Theta^*) }{\omega(1+\alpha(\Theta^*)^2)(1+\alpha\Theta^2)}X_kY +\frac{\alpha\lambda k  (S_k^*)^2  (\Theta+\Theta^*) }{\omega(1+\alpha(\Theta^*)^2)(1+\alpha\Theta^2)}Y^2\right\},
\end{split}\end{equation}
where
\begin{equation}
\begin{split}
X_k=S_k-S_k^*,\quad\quad Y=\Theta-\Theta^*.\\[2mm]
\end{split}\end{equation}

Below we prove that
\begin{equation}\label{x2}\begin{split}
& X_k^2-\frac{\alpha\lambda k S_k^* \Theta^*(\Theta+\Theta^*) }{\omega(1+\alpha(\Theta^*)^2)(1+\alpha\Theta^2)}X_kY  +\frac{\alpha\lambda k  (S_k^*)^2  (\Theta+\Theta^*) }{\omega(1+\alpha(\Theta^*)^2)(1+\alpha\Theta^2)}Y^2\geq 0.
\end{split}
\end{equation}
For this purpose, it suffices to show the following
\begin{equation}\label{eu10}
\begin{split}
\Delta:=&\left[\frac{\alpha\lambda k   S_k^*\Theta^* (\Theta+\Theta^*) }{\omega(1+\alpha(\Theta^*)^2)(1+\alpha\Theta^2)}\right]^2-4\frac{\alpha\lambda k  (S_k^*)^2  (\Theta+\Theta^*) }{\omega(1+\alpha(\Theta^*)^2)(1+\alpha\Theta^2)}\\[2mm]
=&\frac{\alpha\lambda k  (S_k^*)^2  (\Theta+\Theta^*) }{\omega(1+\alpha(\Theta^*)^2)(1+\alpha\Theta^2)}\left[\frac{\alpha\lambda k  (\Theta^*)^2 (\Theta+\Theta^*) }{\omega(1+\alpha(\Theta^*)^2)(1+\alpha\Theta^2)}-4\right]\leq 0.\\[2mm]
\end{split}\end{equation}
Notice that $0\leq \Theta, \Theta^* \leq 1$. Then when $\alpha \leq \frac{2\omega}{n\lambda}$, obviously we have
\begin{equation}\label{eq:alphal}
\frac{\alpha\lambda k  (\Theta^*)^2 (\Theta+\Theta^*) }{\omega(1+\alpha(\Theta^*)^2)(1+\alpha\Theta^2)}-4  \leq \frac{2\alpha \lambda n}{\omega}-4 \leq 0,
\end{equation}
which implies  \eqref{eu10}.

On the other hand, by the inequalities  $a^2+b^2\geq\frac12(a+b)^2$ and $a^2+b^2\geq 2 ab$, we have
\begin{equation}\label{eu6}
\begin{split}
(1+\alpha(\Theta^*)^2)(1+\alpha\Theta^2)=&1+\alpha\big[\Theta^2+(\Theta^*)^2\big]+\alpha^2\Theta^2(\Theta^*)^2\\[2mm]
\geq &1+\alpha\big[\Theta^2+(\Theta^*)^2\big]\\[2mm]
\geq &1+\frac{\alpha}{2}(\Theta+\Theta^*)^2,
\end{split}\end{equation}
and
\begin{equation}\label{eu7}
\begin{split}
I_k^*\leq \frac{\lambda \delta k \Theta^*}{\gamma(\delta+\mu)(1+\alpha(\Theta^*)^2)}\leq \frac{\lambda\delta k}{2\sqrt{\alpha}\gamma(\delta+\mu)},\quad k=1,2,\cdots,n.
\end{split}\end{equation}
It follows from \eqref{eu7} that
\begin{equation}\label{eu8}
\begin{split}
 \Theta^*=&\langle k \rangle^{-1} \sum\limits_{k=1}^n kP(k)I_k^*
 \leq \frac{\lambda\delta}{2\sqrt{\alpha}\gamma(\delta+\mu)\langle k \rangle}\sum\limits_{k=1}^n
k^2P(k)=\frac{R_0}{2\sqrt{\alpha}}.\\
\end{split}\end{equation}
Substituting \eqref{eu6} and \eqref{eu8} into \eqref{eu10} and using the assumption $\alpha\geq \frac{\lambda^2 n^2(R_0)^2}{16 \omega^2}$, we obtain
\begin{equation*}
\begin{split}
 \frac{\alpha\lambda k  (\Theta^*)^2 (\Theta+\Theta^*) }{\omega(1+\alpha(\Theta^*)^2)(1+\alpha\Theta^2)}-4  \leq & \frac{\lambda k (\Theta^*)^2}{\frac{\omega}{2}(\Theta+\Theta^*)}-4 \\[2mm]
 \leq & \frac{2\lambda k \Theta^*}{\omega}-4 \\ \leq &  \frac{\lambda n R_0}{\omega\sqrt{\alpha}}-4 \leq 0,
\end{split}\end{equation*}
which implies \eqref{eu10}. Thus \eqref{x2} holds when $\alpha\geq \alpha_c$ or $\alpha\leq \alpha_l$, so it follows from \eqref{eu11} that
\begin{equation}\begin{split}
\frac{dV}{dt}\leq &-\frac{\Theta}{1+\alpha\Theta^2}\sum\limits_{k=1}^{n}\lambda k a_k(S_k-S_k^*)^2-\sum_{k=1}^{n} \nu\delta a_k (R_k-R_k^*)^2\leq 0.\\[2mm]
\end{split}\end{equation}
Also $\frac{d V}{dt}=0$ if
and only if $S_k=S_k^*$ and $R_k=R_k^*$ for $k=1,2\cdots,n$. According to   the LaSalle's invariant principle\cite{LaSalle1976},
$(S_1^*,R_1^*,\ldots,S_n^*, R_n^*)$ is globally asymptotically stable, so is
 $(S_1^*,I_1^*,\ldots,S_n^*,I_n^*)$.
 The proof is completed.
\end{proof}

\begin{remark}
From Theorem \ref{th2}, we can see under the condition $R_0>1$, the endemic equilibrium is globally asymptotically stable if $\alpha\geq \alpha_c$ or $\alpha \leq a_l$. In the section of numerical experiment, we find that it seems that in the case of $R_0>1$, for any  $\alpha>0$ the system \eqref{eq:nm-si-system} is globally attractive. Also it is worth noting when  $\alpha=0$, the proposed system \eqref{eq:nm-sirs-simplify} can be seen as a generalization to the model proposed in Chen and Sun \cite{Chen2014196}. The global stability of the endemic equilibrium is proved under no assumptions on the vaccination percentage $\mu$ and the recovered rate $\gamma$, which improves the results obtained in  \cite{Chen2014196}. Although the parameter $\alpha>0$ does not affect the epidemic threshold, we can show easily that
\begin{equation}\label{eq:important}
  \limsup_{t\to \infty} I_k(t)<\frac{\lambda k}{2\gamma \sqrt{\alpha}},
\end{equation}
which together with equation \eqref{eu7}, we can see that a larger $\alpha$ can accelerate the extinction of the disease and reduce the level of disease, which is further verified in the numerical results.
\end{remark}

\section{Numerical results}
In this section, we will give some numerical simulations to illustrate the theoretical findings. Our simulations are based on the BA network with maximum degree $n=500$. The degree distribution $P(k)=ck^{-m}$ , and the constant $c$ satisfying $\sum_{k=1}^{n}P(k)=1$.

In all of the experiments below, we fix $m=3, \lambda=0.01$ and  $\delta=0.02$.  In Fig.\ref{fig1}, we choose $\gamma=0.01$ and $\mu=0.2$. By simple computations, we can derive $R_0\approx 0.3759< 1$. The time series of infected nodes with degree $k=100,200,\cdots, 500$ are shown in Fig.\ref{fig1a}($\alpha=0.1$) and Fig.\ref{fig1b}($\alpha=10$). The initial values are $S_k(0)=0.3, I_k(0)=0.1, R_k(0)=0.6$ for any degree $k$.
It can be seen that the disease eventually becomes extinct when $R_0<1$, which implies
that the disease-free equilibrium is stable. Meanwhile it should be noted that, although for both choices of $\alpha>0$, the density of the infected node with a bigger degree experiences a peak before it leads to the disease-free equilibrium, the larger $\alpha$ is, the lower the peak level is, which means that large $\alpha$ can effectively degrade the epidemic peak when it breaks out initially.
To show the global stability of the disease-free equilibrium, we choose 10 different initial values to plot the time series of $I_{400}(t)$ in Fig.\ref{fig2}. Obviously all trajectories converge to the trivial equilibrium and this supports the global stability of disease-free equilibrium. Also as is shown in Fig.\ref{fig2b}, when $R_0<1$, the larger $\alpha>0$ is , the faster the disease dies out.

\begin{figure}
  \centering
  \subfigure[$\alpha=0.1$]{
    \label{fig1a} 
    \includegraphics[width=0.48\textwidth]{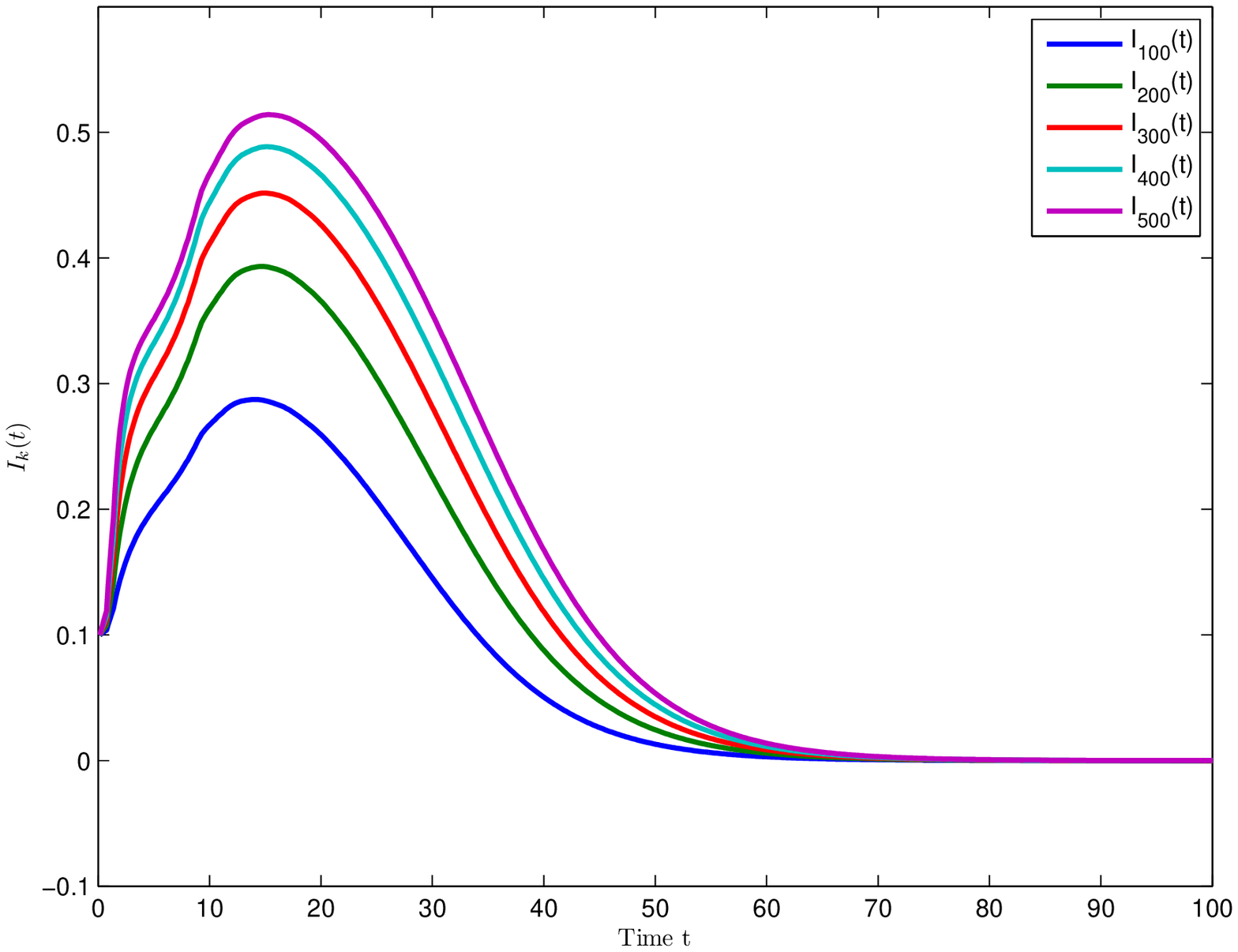}}
  \hfill
  \subfigure[$\alpha=10$]{
    \label{fig1b} 
    \includegraphics[width=0.48\textwidth]{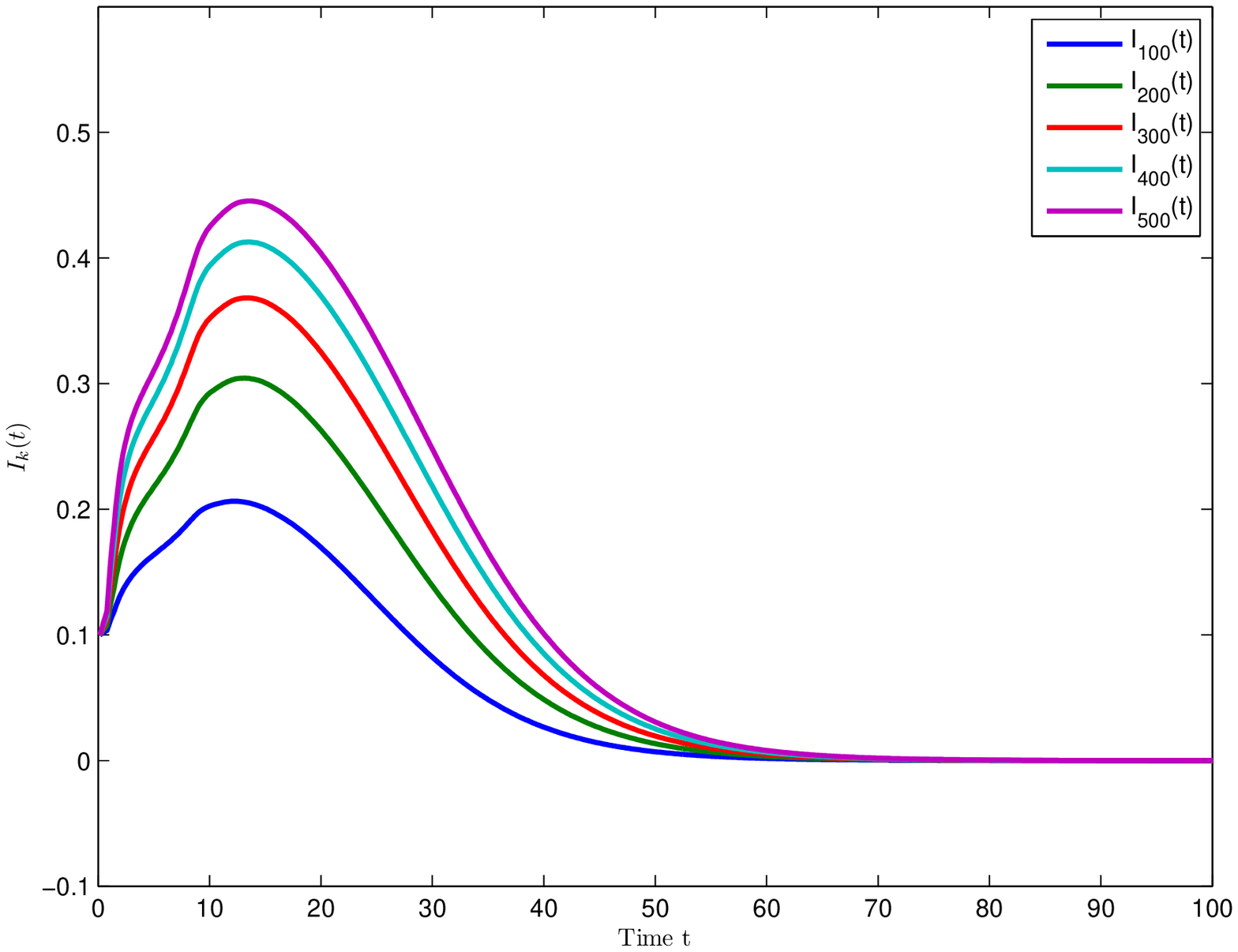}}

  \caption{The time series of infected nodes in \eqref{eq:nm-si-system} with $R_0 < 1$. }
  \label{fig1} 
\end{figure}

\begin{figure}
  \centering
  \subfigure[$\alpha=0.1$]{
    \label{fig2a} 
    \includegraphics[width=0.48\textwidth]{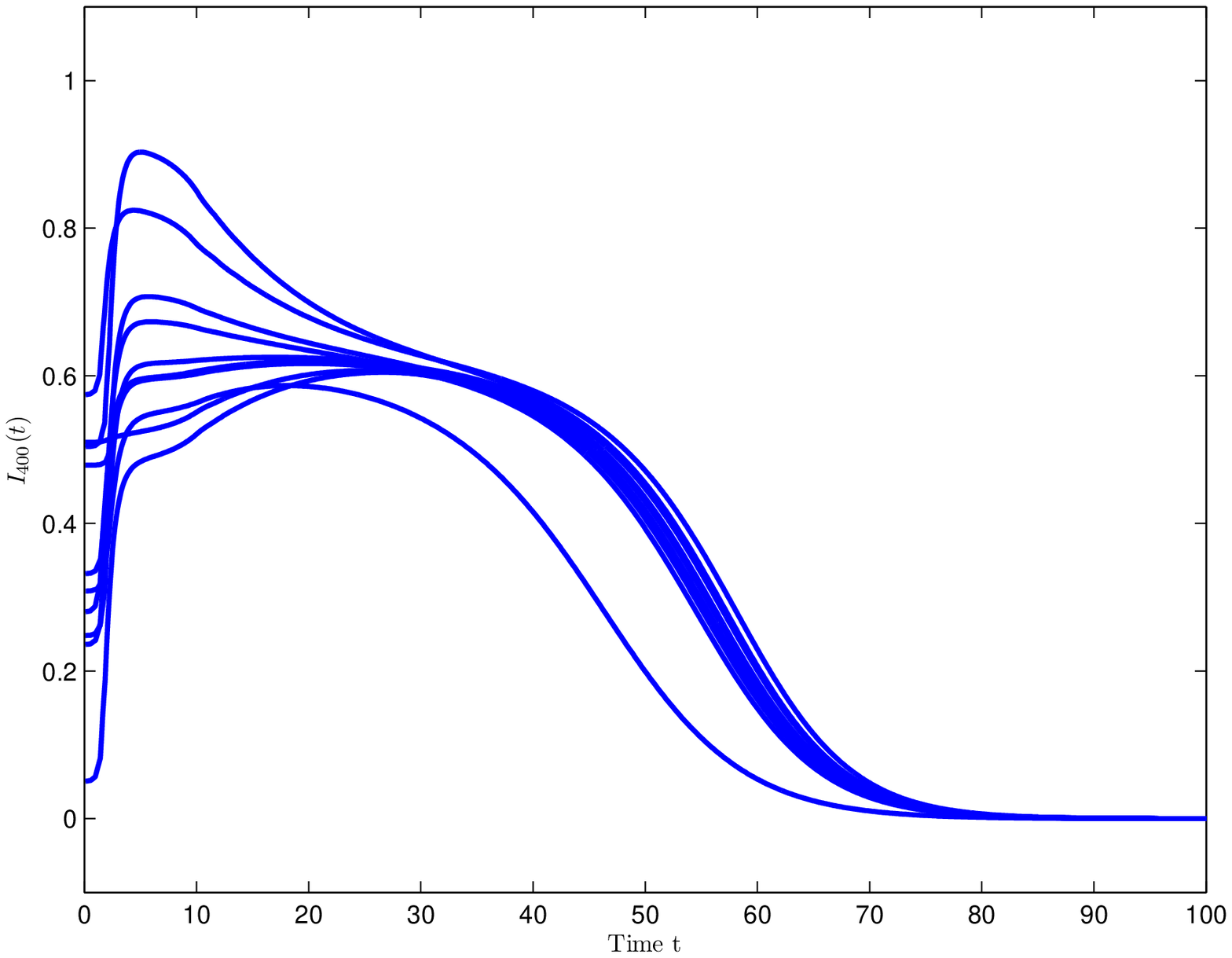}}
  \hfill
  \subfigure[$\alpha=10$]{
    \label{fig2b} 
    \includegraphics[width=0.48\textwidth]{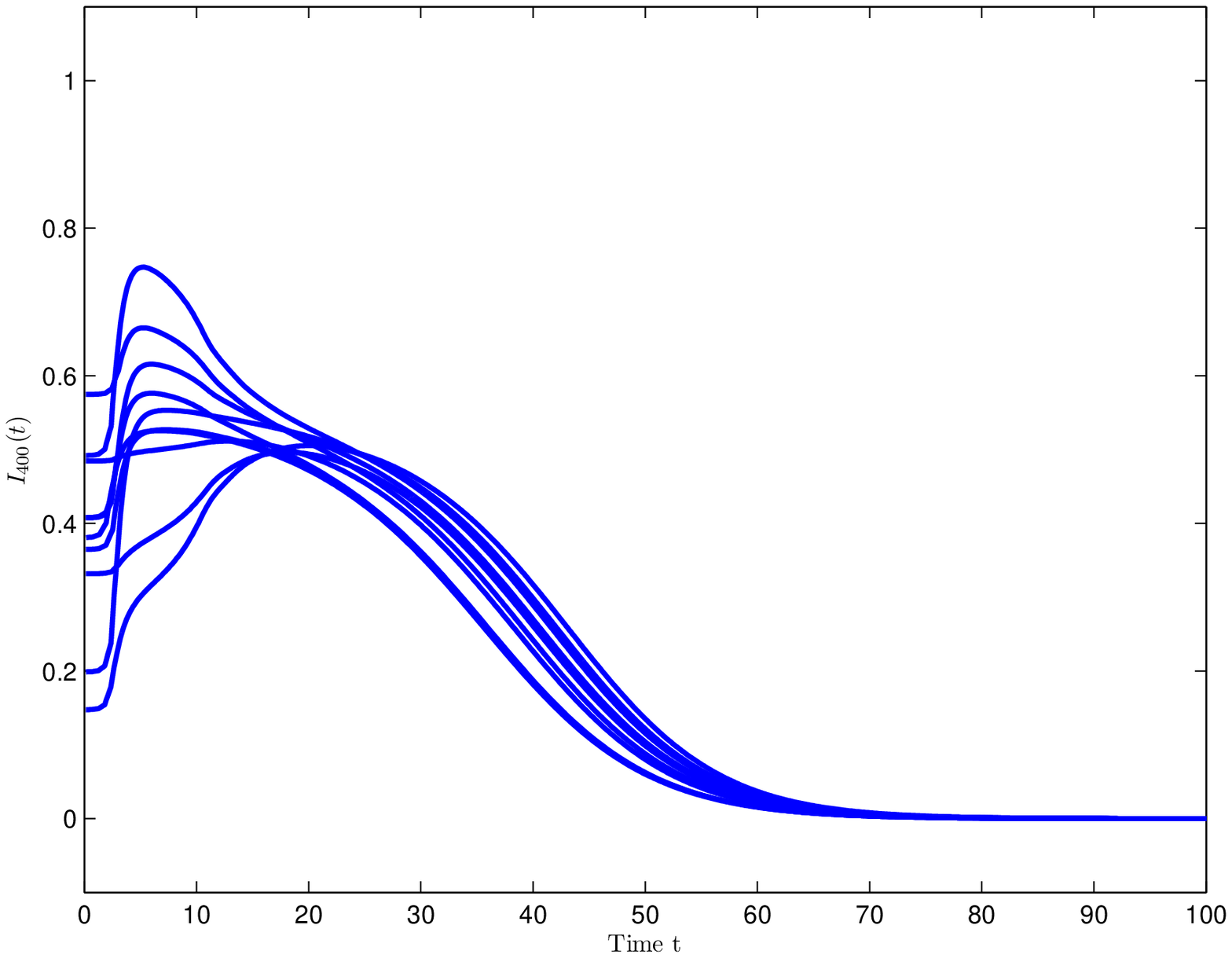}}
  \caption{The time series of $I_{400}(t)$ with 10 different initial values and $R_0<1$. }
  \label{fig2} 
\end{figure}

In Fig.\ref{fig3} and Fig.\ref{fig3m}, we choose $\mu=0.02$ and $\gamma=0.01$. Obviously we have $\mu>\gamma$. By simple computations, we can derive $R_0\approx 2.0673>1$. Using the equation \eqref{eq:B5}, we obtain $\omega=0.0040$. Thus we have $\alpha_l=0.0016$ and $\alpha_c=417349.00$. We choose $\alpha=10^{-3}$ and $10^6$ respectively so that the condition $\alpha<\alpha_l$ or $\alpha>\alpha_c$ in Theorem \ref{th2} is satisfied. In Fig.\ref{fig4} and Fig.\ref{fig4m}, we choose $\mu=0.001$ and $\gamma=0.01$. Thus we have $\mu < \gamma$ and $R_0\approx 3.9377>1$. Accordingly, the parameters in Theorem \ref{th2} are computed as $\omega=0.0004, \alpha_l=0.0002, \alpha_c\approx 1.3\times 10^8$. So in this case we choose $\alpha=10^{-4}<\alpha_l$ and $\alpha=10^9>\alpha_c$ respectively. As can be see from Fig.\ref{fig3} and Fig.\ref{fig4}, in the case of $R_0>1$ the disease will persist on a positive steady level. Fig.\ref{fig3m} and Fig.\ref{fig4m} display the evolutions of $I_{400}(t)$ for a set of initial
conditions. In both cases, as can be seen from Fig.\ref{fig3} to Fig.\ref{fig4m}, the endemic equilibrium $E^*$ is globally asymptotically stable whenever $R_0>1$ and $\alpha>\alpha_c$ or $\alpha<\alpha_l$, which is in agreement with Theorem \ref{th2}. Again we can see in Fig.\ref{fig3} and Fig.\ref{fig4}, when the disease is endemic, the densities of the infected nodes decreases as $\alpha$ increases, which means that a larger $\alpha$ can accelerate the extinction of the disease and reduce the level of disease.

%

\begin{figure}
  \centering
  \subfigure[$\alpha=10^{-3}$]{
    \label{fig3a} 
    \includegraphics[width=0.48\textwidth]{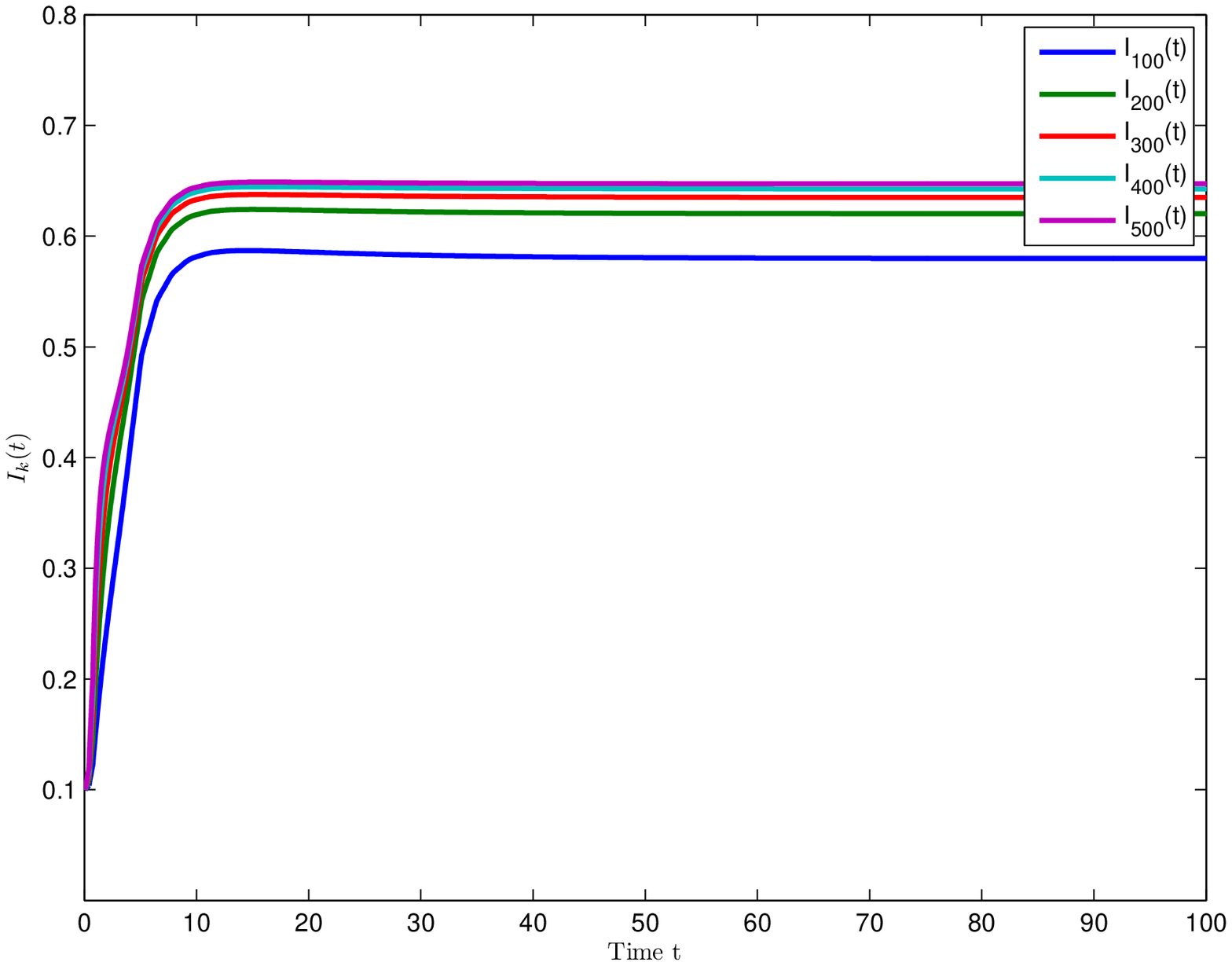}}
  \hfill
  \subfigure[$\alpha=10^6$]{
    \label{fig3b} 
    \includegraphics[width=0.48\textwidth]{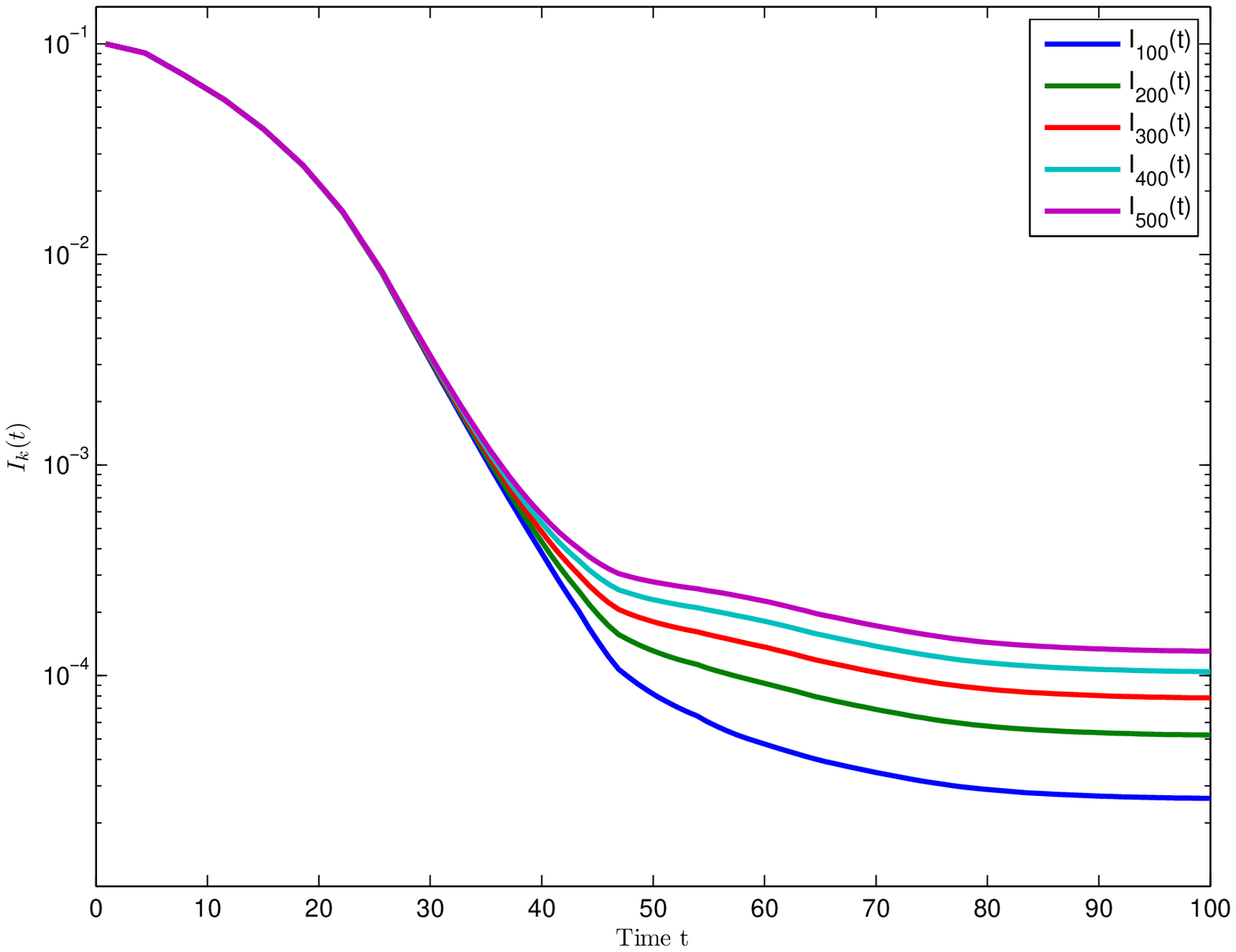}}

  \caption{The time series of infected nodes in \eqref{eq:nm-si-system} with $R_0 > 1$ and $\mu > \gamma$. }
  \label{fig3} 
\end{figure}

\begin{figure}
  \centering
  \subfigure[$\alpha=10^{-3}$]{
    \label{fig3am} 
    \includegraphics[width=0.48\textwidth]{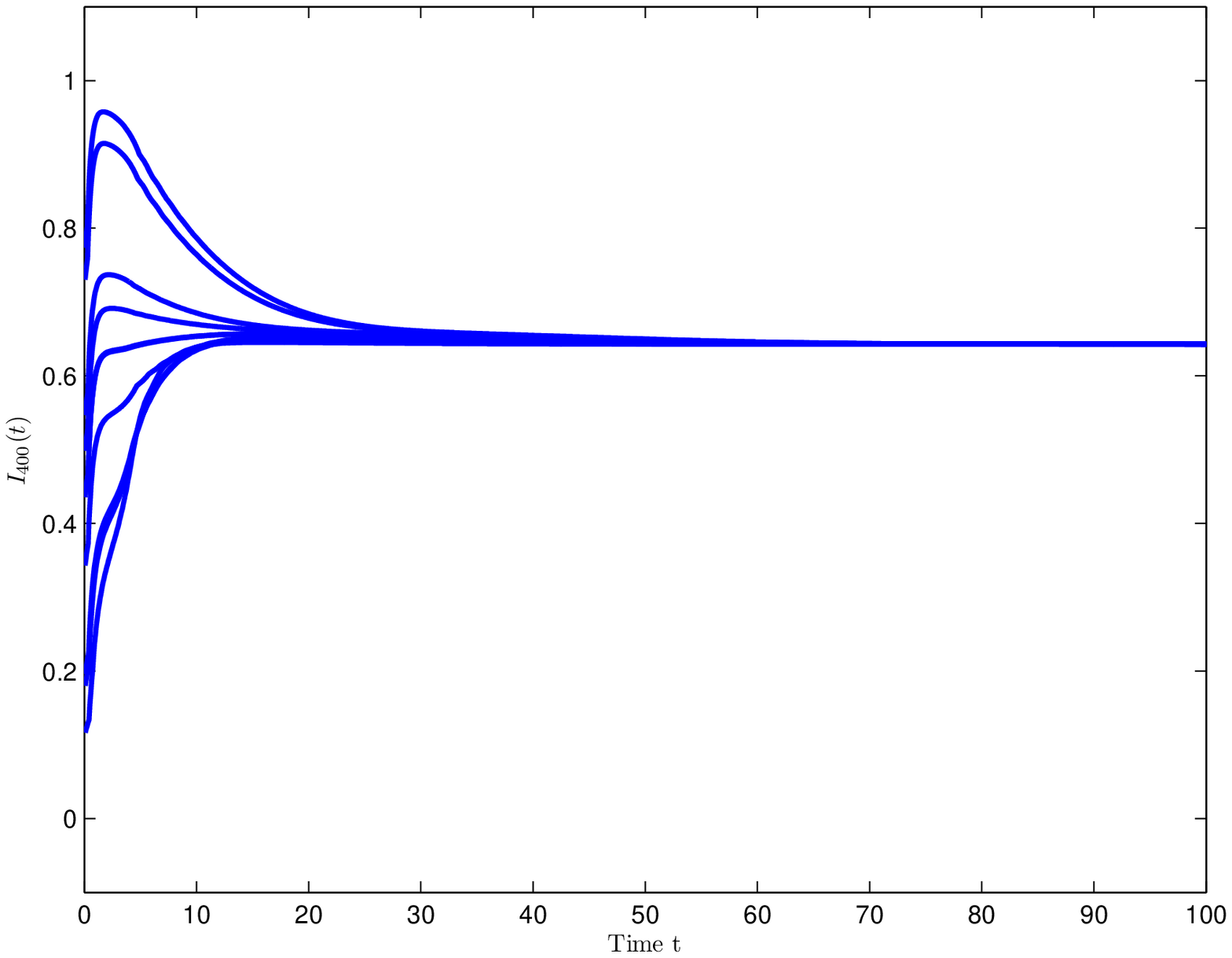}}
  \hfill
  \subfigure[$\alpha=10^6$]{
    \label{fig3bm} 
    \includegraphics[width=0.48\textwidth]{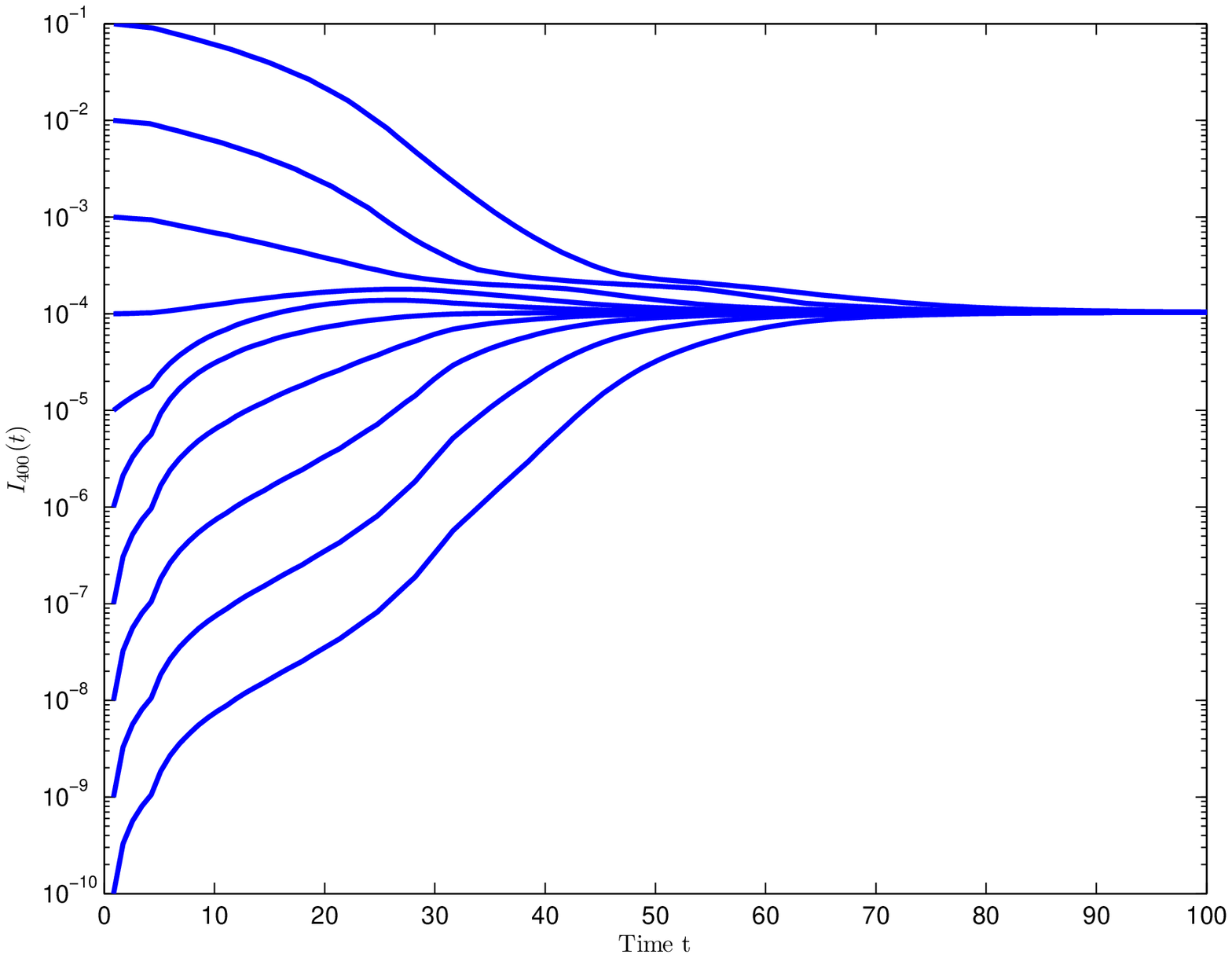}}

  \caption{The time series of $I_{400}(t)$ with 10 different initial values when $R_0>1$ and $\mu>\gamma$. }
  \label{fig3m} 
\end{figure}

%

\begin{figure}
  \centering
  \subfigure[$\alpha=10^{-4}$]{
    \label{fig4a} 
    \includegraphics[width=0.48\textwidth]{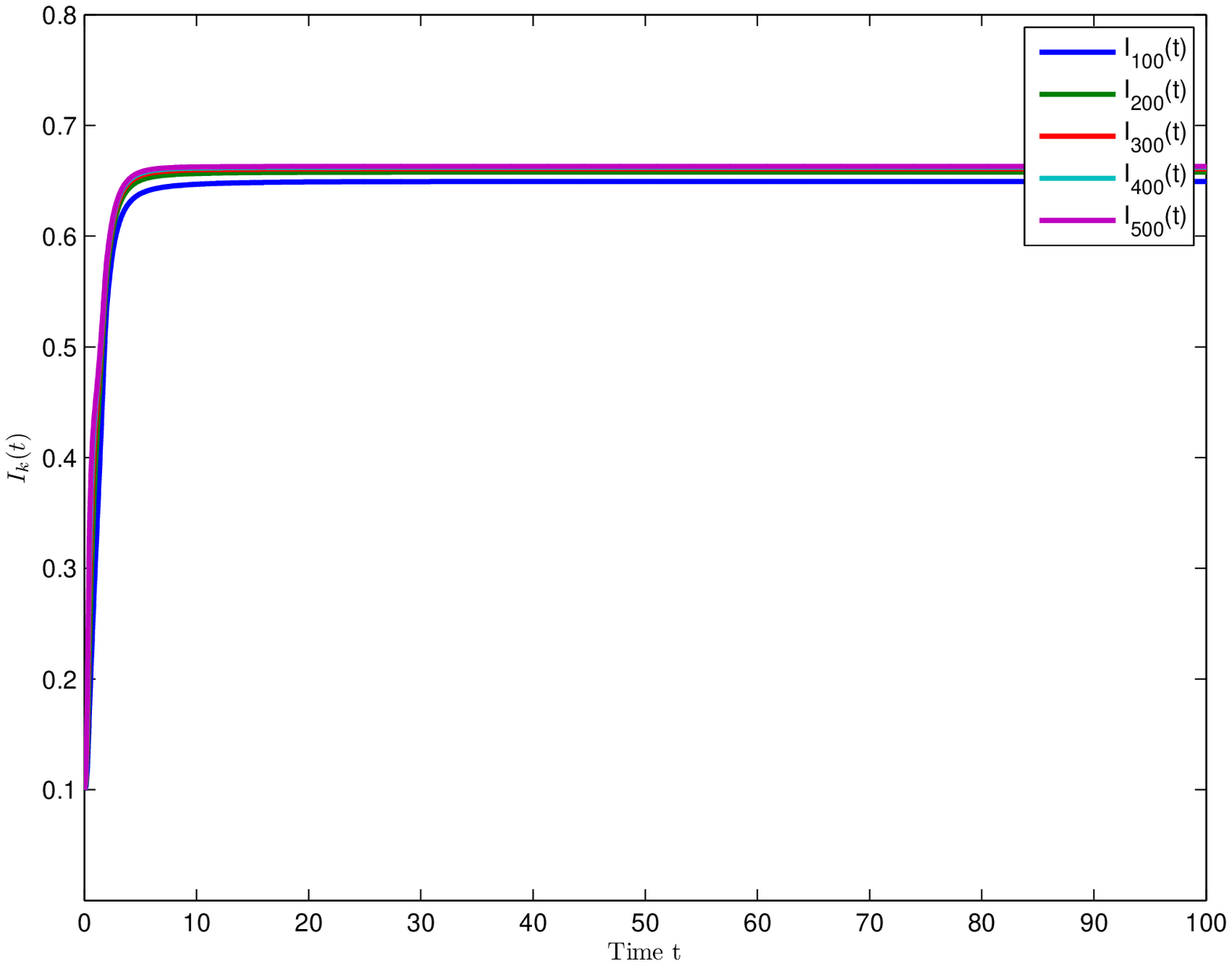}}
  \hfill
  \subfigure[$\alpha=10^9$]{
    \label{fig4b} 
    \includegraphics[width=0.48\textwidth]{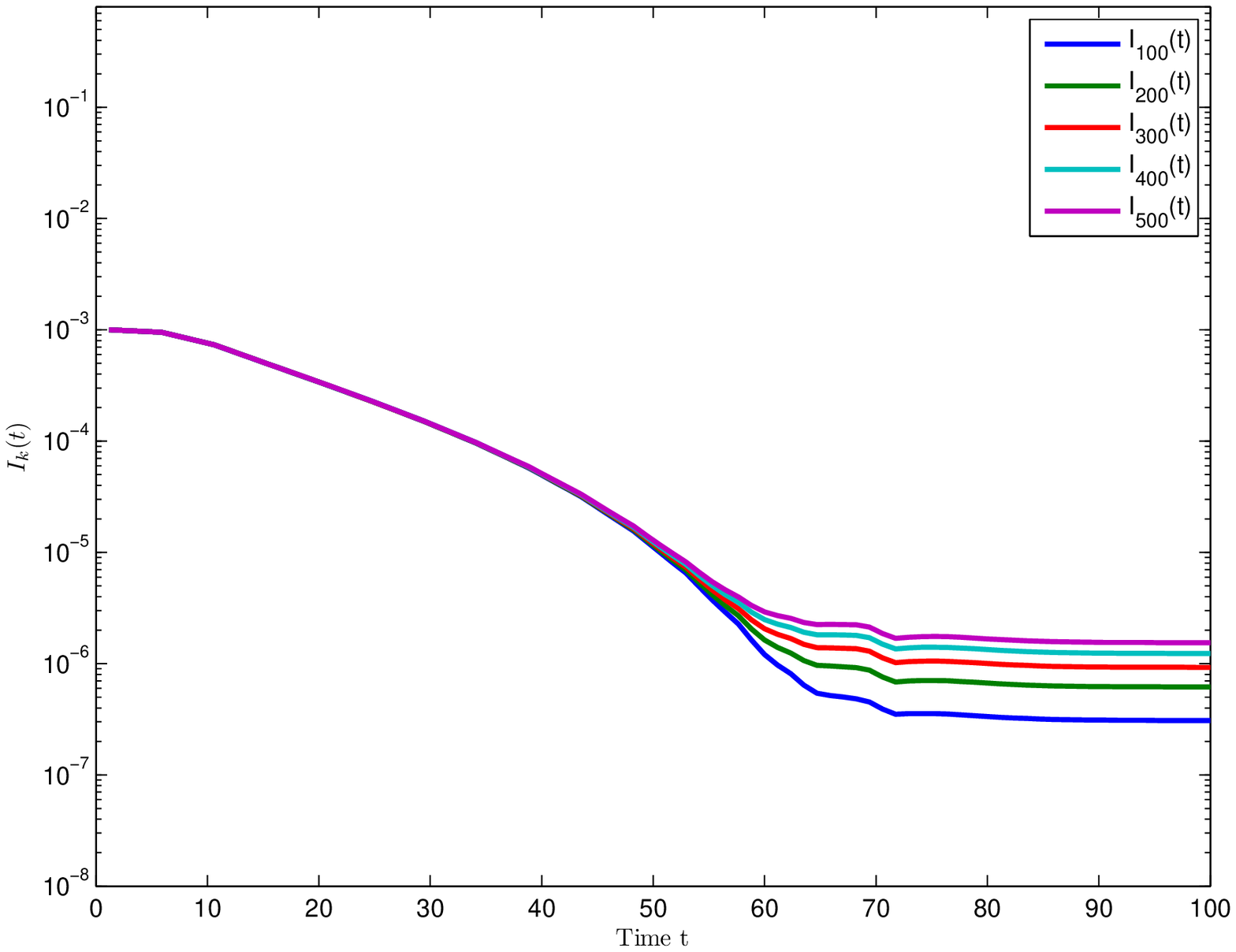}}

  \caption{The time series of infected nodes in \eqref{eq:nm-si-system} with $R_0 > 1$ and $\mu < \gamma$. }
  \label{fig4} 
\end{figure}

\begin{figure}
  \centering
  \subfigure[$\alpha=10^{-4}$]{
    \label{fig4am} 
    \includegraphics[width=0.48\textwidth]{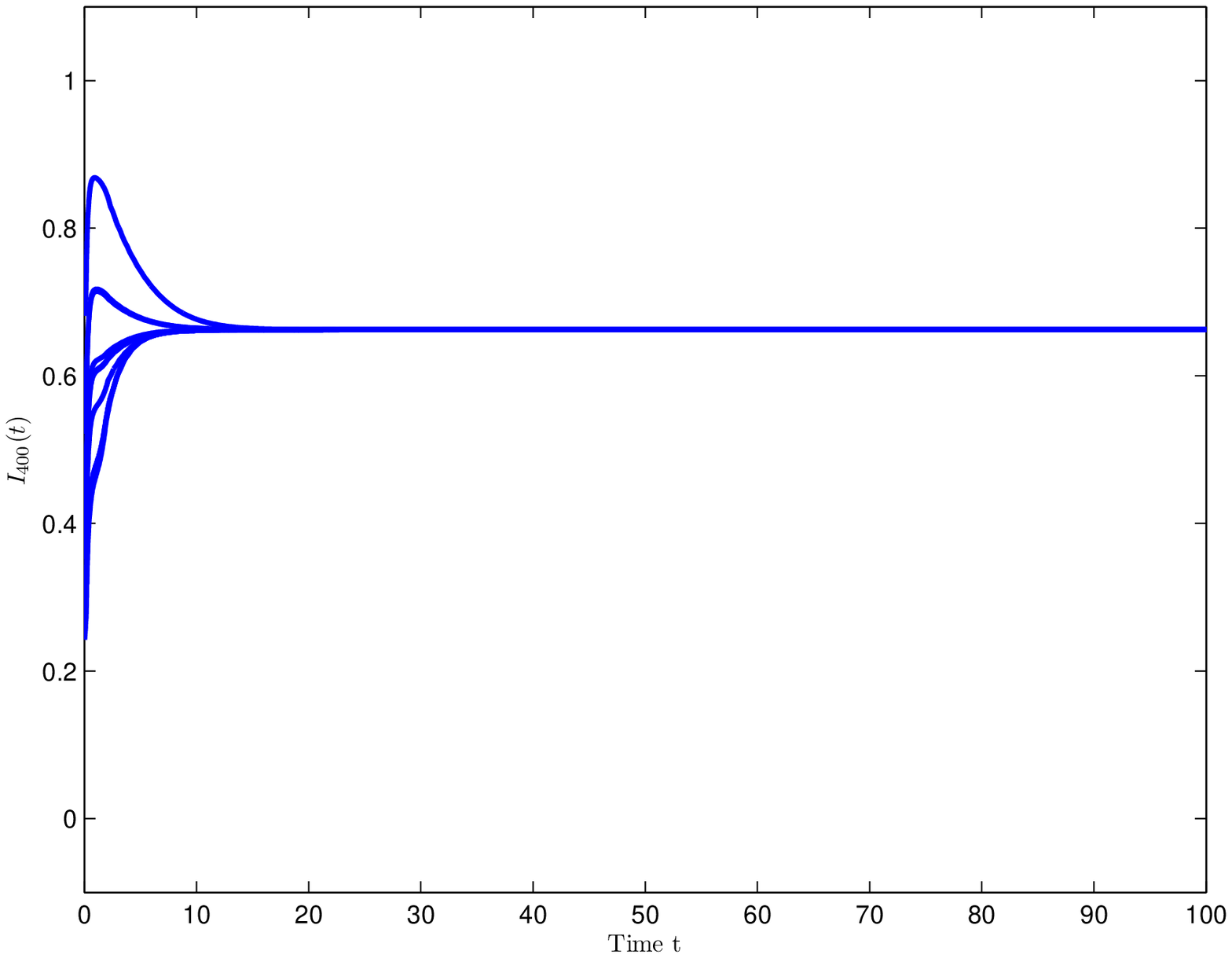}}
  \hfill
  \subfigure[$\alpha=10^9$]{
    \label{fig4bm} 
    \includegraphics[width=0.48\textwidth]{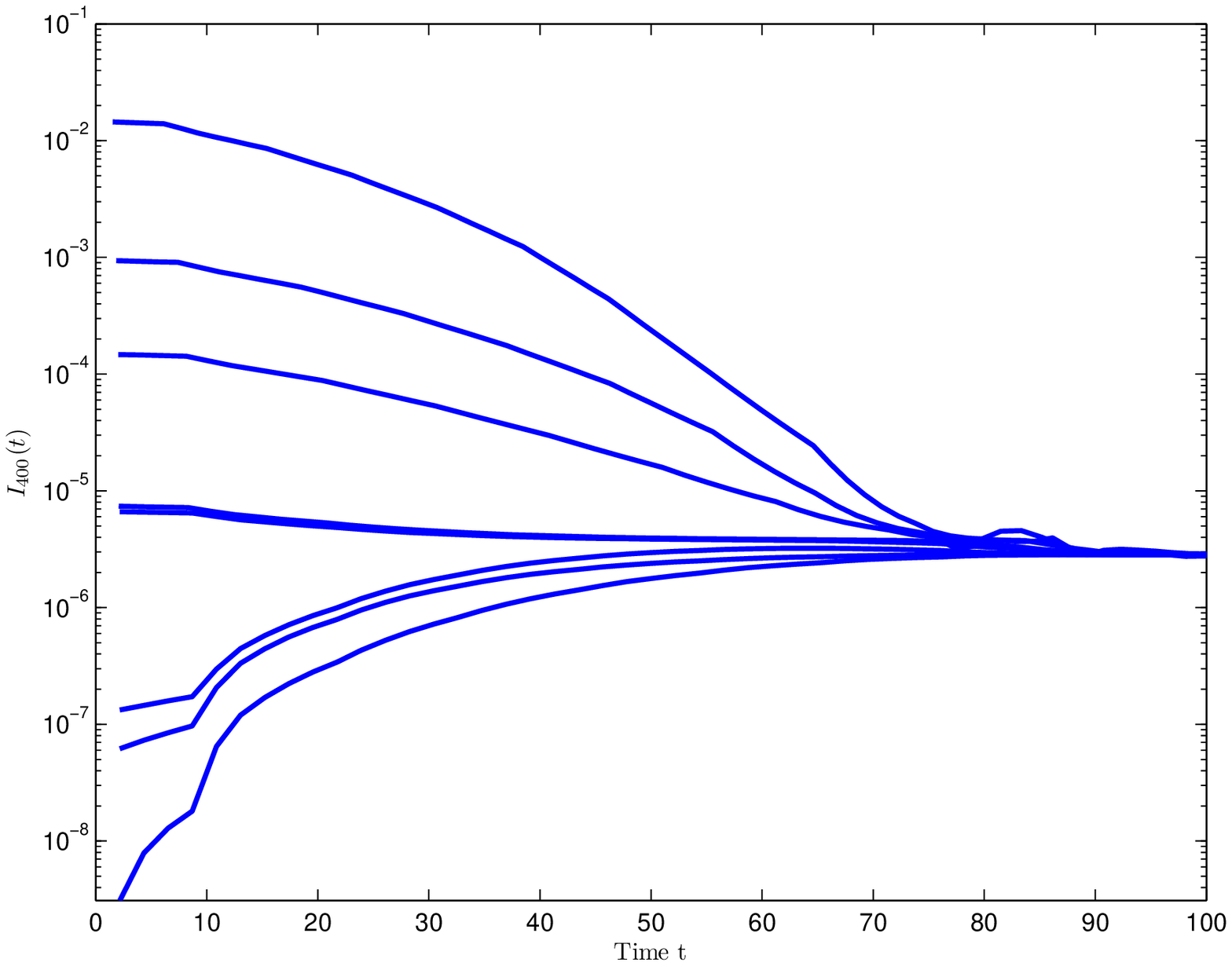}}

  \caption{The time series of $I_{400}(t)$ with 10 different initial values when $R_0>1$ and $\mu<\gamma$. }
  \label{fig4m} 
\end{figure}

When $R_0>1$, Theorem \ref{th2} shows the endemic equilibrium $E^*$ is globally asymptotically stable when $\alpha$ is sufficiently large or small enough. However, as is shown in Fig.\ref{fig5}, for any $\alpha>0$ it seems that $E^*$ is indeed attractive when $R_0>1$. However, the rigourous proof to this conclusion is still mathematically difficult which will be considered later.

%
%
%

\begin{figure}
  \centering
  \subfigure{
    \label{fig5a} 
    \includegraphics[width=0.48\textwidth]{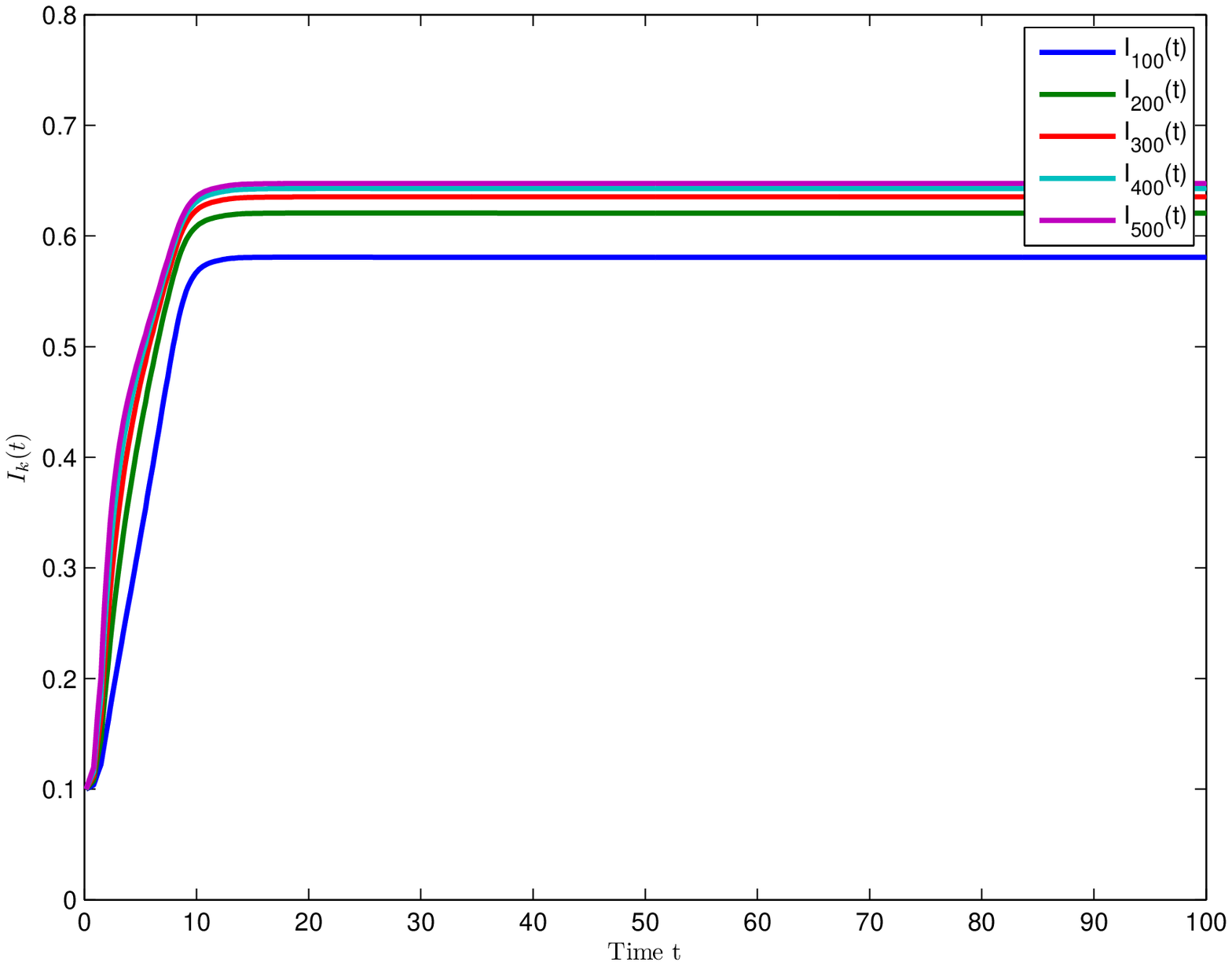}}
  \hfill
  \subfigure{
    \label{fig5b} 
    \includegraphics[width=0.48\textwidth]{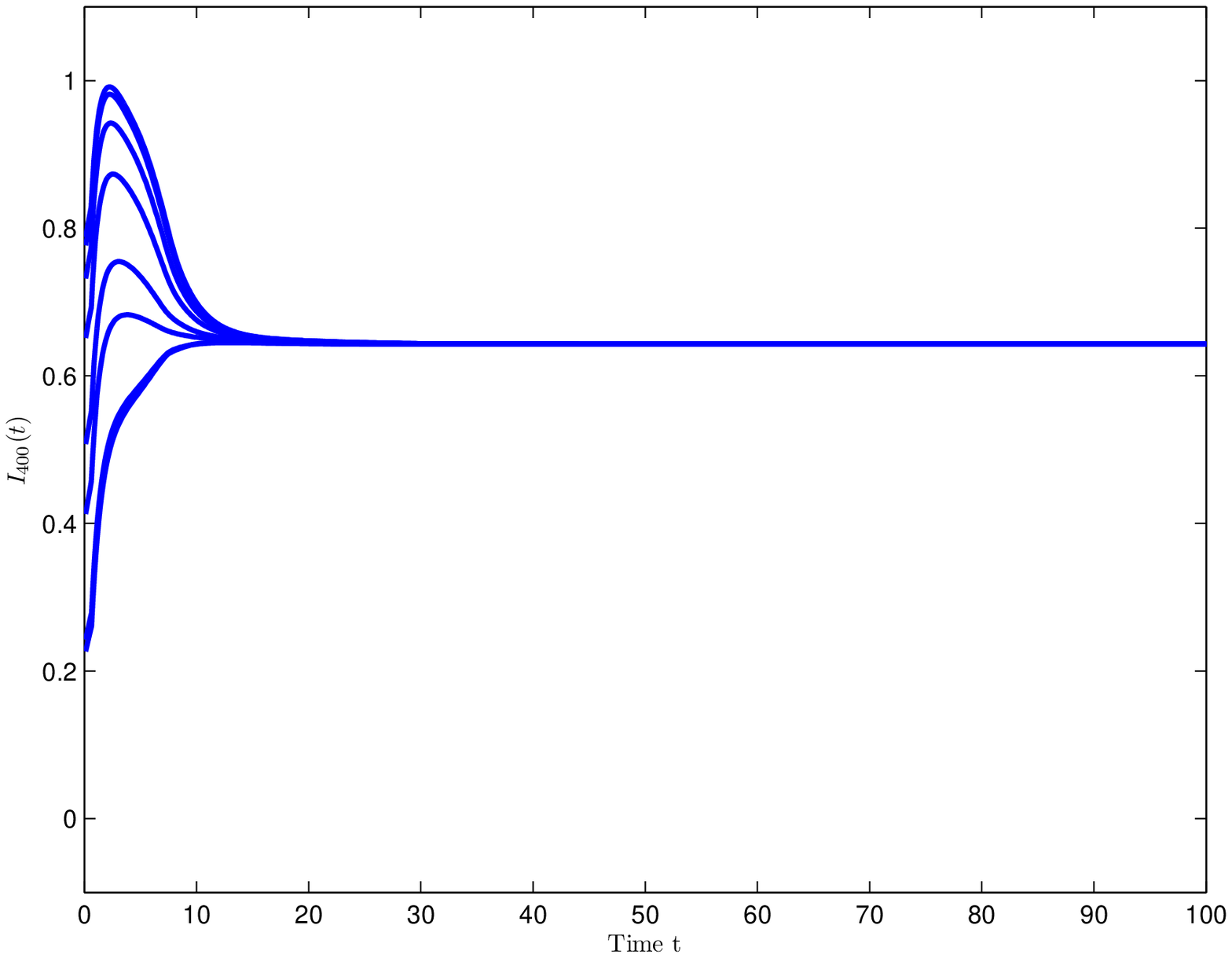}}

  \caption{The time series of infected nodes in \eqref{eq:nm-si-system} with $R_0 > 1$ and $\alpha=10$. }
  \label{fig5} 
\end{figure}

\section{Conclusion}
In this paper, we have discussed an SIRS epidemic model with vaccination and nonomontone incidence rate on complex networks. The nonlinear incidence rate can be used to interpret the psychological effect, namely, the incidence rate would decrease at high infective levels due to the quarantine of infected individuals or the protection measures taken by the susceptible ones. Although the parameter $\alpha$ does not affect the epidemic threshold, it plays a role in weakening the spreading of disease, as can be seen in equations \eqref{eu7} and \eqref{eq:important}. We have shown by Lyapunov function that without additional conditions on the constants $\mu$ and $\gamma$, the endemic equilibrium of system \eqref{eq:nm-sirs-simplify} is globally asymptotically stable, thus the disease becomes endemic. The results can be viewed as an important supplement to the result in \cite{Chen2014196}. Furthermore, numerical simulations are done. As has been seen, the simulations verify the globally asymptotical stability of $E^*$ and controlling effect of the inhibitory factor $\alpha>0$ on reducing the disease level.

\vskip0.4cm
\section*{Acknowledgments}
The research was supported  in part by the NSFC (grants 61002039, 61572018, 11571062),
the Program for Liaoning Excellent Talents in University (grant
LJQ2013124)
  and  the Fundamental Research Fund  for the Central Universities  (grants DC201502050404,  DC201502050202).
\par
\begin{appendices}
  \section{Proof of Lemma \ref{lm:1}}
Since $S_k(0)>0$ for $k=1,2,\cdots,n$, by continuity, we have that for any $k$, there exists
some $t_k>0$ such that $S_k(t)>0$  for any $t\in (0, t_k)$. Thus, the set
${\cal E}_k:=\{\tau>0; S_k(t))>0,\forall t \in (0, \tau)\}$ is not empty for
any $k=1,2,\cdots,n$. Let $\alpha_k=\sup{\cal E}_k$ for $k=1,2,\cdots,n$.
Then $\alpha_k>0$ for $k=1,2,\cdots,n$.

We will show $S_k(t)>0$ for all $t>0$ and all $k=1,2,\cdots,n$. To prove this, it suffices to show that  $\alpha_{k}=\infty$ for $k=1,2,\cdots,n$. Suppose, on the contrary, that $\alpha_m<\infty$ for some $m \in
\{1,2,\cdots,n\}$. Then $S_m(t)>0$ for all $t\in (0, \alpha_m)$. By continuity, there must have $S_m(\alpha_m)=0$. In the following, we will show this is not true.

Firstly, we will prove $\Theta(t)>0$ for all $t>0$. From the second equation of
\eqref{eq:nm-si-system}, we have
\begin{equation}\label{e030}
\begin{split}
&\frac{d\Theta(t)}{dt}= \left[\frac{\lambda}{\langle
k\rangle}\sum\limits_{k=1}^nk^2P(k)\frac{S_k(t)}{1+\alpha \Theta^2(t)}-\gamma\right]\Theta(t)=:X(t)\Theta(t),
\end{split}
\end{equation}
which gives
\begin{equation*}
\begin{split}
 \Theta(t)=\Theta(0) \exp\left\{\int_0^t X(s) ds\right\}>0,\quad \forall t>0.
\end{split}
\end{equation*}

%
\newcommand{\bs}{\delta}
\newcommand{\ds}{\mu}
\newcommand{\gs}{\frac{\Theta(s)}{1+\alpha\Theta^2(s)}}
Thus, for any $t \in (0,\alpha_m)$, by the second equation of
\eqref{eq:nm-si-system}, we also obtain that
\begin{equation}
\begin{split}\frac{d}{dt}\left(\exp\{\gamma t\}I_m(t)\right)
=\lambda m
S_{m}(t)\frac{\Theta(t)}{1+\alpha\Theta^2(t)}\exp\{\gamma t\}>0,\end{split}
\end{equation}
which means that  $I_m(t)\geq 0$ for all $t\in (0,\alpha_m]$. 

Next we will prove that $S_m(t)+I_m(t)< 1$ for
all $t\in (0,\alpha_m]$. Adding the two equations in \eqref{eq:nm-si-system}
yields
\begin{equation}\label{e025}
\begin{split}
 \frac{d (S_k+I_k)}{dt}=\bs-(\bs+\mu)(S_k+I_k)+(\mu-\gamma) I_k,\quad k=1,2,\cdots,n,
\end{split}
\end{equation}
which is equivalent to
\begin{equation}\label{e026}
\begin{split}
 \frac{d (S_k+I_k)}{dt}=\bs-(\bs+\gamma)(S_k+I_k)+(\gamma-\mu) I_k,\quad k=1,2,\cdots,n.
\end{split}
\end{equation}
If $\mu\leq \gamma$,
we derive from \eqref{e025} with $k=m$ that
\begin{equation*}
\begin{split}
 \frac{d (S_m+I_m)}{dt}\leq & \bs-(\bs+\mu)(S_m+I_m)\\
 < &(\bs+\mu)(1-S_m-I_m),\quad \forall t \in (0, \alpha_m].
\end{split}
\end{equation*}
This implies that
$\frac{d}{dt}\left(\exp\{(\bs+\mu)t\}(1-S_{m}(t)-I_m(t))\right)>0$ for all $
t \in (0, \alpha_m]$, which together with $S_m(0)+I_m(0)\leq 1$ yields that
$S_m +I_m  < 1$ on $(0, \alpha_m]$. While if $\mu>\gamma$, we derive from
\eqref{e026} that
\begin{equation*}
\begin{split}
 \frac{d (S_m+I_m)}{dt}\leq & \bs-(\bs+\gamma)(S_m+I_m)\\
<&(\bs+\gamma)(1-S_m-I_m),\quad \forall t \in (0, \alpha_m].
\end{split}
\end{equation*}
This gives that $\frac{d}{dt}\left(\exp\{(\bs+\gamma)t\}(1-S_m-I_m)\right)>0$
on $(0, \alpha_m]$, which also implies that $S_m+I_m <1$ on $(0, \alpha_m]$.

Thus, it follows from the first equation of \eqref{eq:nm-si-system} that
\begin{equation*}
\begin{split}
&\frac{d}{dt}\left(\exp\left\{\int_0^t\left(\ds+\lambda m\gs\right)ds\right\}S_m(t)\right)\bigg{|}_{t=\alpha_m}\\
&=\bs(1-S_{m}(\alpha_m)-I_{m}(\alpha_m))\exp\left\{\int_0^{\alpha_m}\left(\ds+\lambda m\gs\right)ds\right\}>0,
\end{split}
\end{equation*}
which means  $S_m(t)<S_m(\alpha_m)=0$ for $t\in (\alpha_m-\epsilon,\alpha_m) \subset (0,\alpha_m)$, where $\epsilon$ is an arbitrarily positive constant.  This is apparently a contradiction.  Thus we have $\alpha_{k}=\infty$ for
$k=1,2,\cdots,n$. In conclusion, we have $S_k(t)>0 $
 for all $t>0$.

Then in the same way, we
get $I_k(t)>0$ from the second equation of \eqref{eq:nm-si-system} and $S_k(t)+I_k(t)<1$ from \eqref{e025} and \eqref{e026} for all $t>0$. Summing up the above arguments and combining the fact
$S_k(t)+I_k(t)+R_k(t)=1$, we have $0<S_k(t),I_k(t),R_k(t)<1$ and $0<\Theta(t)<1$ for all $t>0$
and all $k=1,2,\cdots,n$. This completes the proof.

 \section{Proof of equation \eqref{eq:inequ-Fk(M)}}
Denote $X=R_k-R_k^*$ and $Y=S_k-S_k^*$. Then following the expression of $F_k(m)$ in \eqref{eq:Fk(M)}, we have
\begin{equation}\label{eqaa1}
\begin{split}
F_k(m)& =(\gamma+\delta)mX^2+[(\gamma-\mu)m-\delta]XY+\mu Y^2\\
      & =(\gamma+\delta)m\left[X+\frac{(\gamma-\mu)m-\delta}{2(\gamma+\delta)m}Y\right]^2+2\omega(m)Y^2 \\
      & =\mu \left[Y+\frac{(\gamma-\mu)m-\delta}{2\mu}X\right]^2+2\nu(m)X^2,
\end{split}
\end{equation}
where $\omega(m)=-\frac{\triangle(m)}{8(\gamma+\delta)m}$ and $\nu(m)=-\frac{\triangle(m)}{8\mu}$ with
\begin{equation}\label{eqaa2}
  \begin{split}
    \triangle(m) & =[(\gamma-\mu)m-\delta]^2-4\mu(\gamma+\delta)m \\
       & = (\gamma-\mu)^2 m^2-[2\delta(\gamma-\mu)+4\mu(\gamma+\delta)]m+\delta^2.
  \end{split}
\end{equation}
We now show there exists $m^*>0$ such that $\triangle(m^*)<0$.  We prove this by two parts.
\begin{enumerate}
  \item If $\mu=\gamma$, we can choose $m^*=\frac{\delta}{4\mu}>0.$ It is easy to see $\triangle(m^*)=-\gamma\delta < 0.$
  \item If $\mu \neq \gamma$, let $m^*=\frac{\delta\gamma+\mu\delta+2\mu\gamma}{(\gamma-\mu)^2}>0$, obviously we have

      \begin{equation}\label{eq11}
      \triangle(m^*)=-\frac{4\mu\gamma(\gamma+\delta)(\delta+\mu)}{(\gamma-\mu)^2}<0.
      \end{equation}

\end{enumerate}

Thus the above claim is proved. So we have $\omega(m^*)>0$ and $\nu(m^*)>0$, from which we have
\begin{equation}\label{eqb1}
  F_k(m^*) \geq \omega(m^*) (S_k-S_k^*)^2+\nu(m^*)(R_k-R_k^*)^2,
\end{equation}
with
\begin{equation}\label{eq:B5}
    \omega(m^*)=-\frac{\triangle(m^*)}{8(\gamma+\delta)m^*},
\end{equation}
and
\begin{equation}\label{eq:B6}
   \nu(m^*)=-\frac{\triangle(m^*)}{8\mu}.
\end{equation}

\end{appendices}

%

\bibliographystyle{abbrv}
\bibliography{nmsirs}
\end{document}